\documentclass[acmsmall,screen,authorversion]{acmart}

\setcopyright{rightsretained}
\acmPrice{}
\acmDOI{10.1145/3571740}
\acmYear{2023}
\copyrightyear{2023}
\acmSubmissionID{popl23main-p600-p}
\acmJournal{PACMPL}
\acmVolume{7}
\acmNumber{POPL}
\acmArticle{68}
\acmMonth{1}
\received{2022-07-07}
\received[accepted]{2022-11-07}

\startPage{1}

\usepackage{booktabs}   \usepackage{subcaption} 

\usepackage[utf8]{inputenc}

\usepackage{xspace}

\usepackage{listings}
\lstnewenvironment{pseudocode}{\lstset{language=Pascal}}{}

\usepackage{wrapfig}

\newcommand{\para}[1]
{\subsubsection*{\textnormal{\bf #1}}}
  
\newcommand{\mathsc}[1]{{\normalfont\textsc{#1}}}

\def\etal.{et\penalty50\ al.} 

\renewcommand{\iff}{\mathrel{\mathrm{\ iff\ }}}

\newcommand{\eqdef}{\mathrel{\stackrel{{\scriptscriptstyle\mathrm{def}}}{=}}}

\newtheorem{theorem}{Theorem}
\newtheorem{proposition}{Proposition}
\newtheorem{lemma}{Lemma}
\newtheorem{corollary}{Corollary}
\newtheorem{definition}{Definition}
\newtheorem*{notation}{Notation}
\newtheorem{fact}{Fact}

\newcommand{\setof}[1]{\{#1\}}
\newcommand{\powerset}{\raisebox{.15\baselineskip}{\Large\ensuremath{\wp}}}

\newcommand{\dom}{\mathrm{dom}}
\newcommand{\codom}{\mathrm{codom}}

\newcommand{\REL}{\ensuremath{\operatorname{REL}}\xspace}
\newcommand{\Rel}{\REL}

\newcommand{\LoI}{\ensuremath{\operatorname{LoI}}\xspace}
\newcommand{\LOI}{\LoI}

\newcommand{\infoleq}{\precsim} \newcommand{\pinfoleq}{\LOCIvariant{\infoleq}} 

\newcommand{\ER}{\ensuremath{\mathrm{ER}}}
\newcommand{\PRE}{\ensuremath{\mathrm{PRE}}}
\newcommand{\LOIvariant}[1]{\mathrel{#1_\text{\tiny \LOI}}}
\newcommand{\LOIleq}{\LOIvariant{\sqsubseteq}}
\newcommand{\LOIgeq}{\LOIvariant{\sqsupseteq}}
\newcommand{\LOIjoin}{\LOIvariant{\sqcup}}
\newcommand{\LOImeet}{\LOIvariant{\sqcap}}
\newcommand{\sLOIbigJoin}{\textstyle{\bigsqcup_{\text{\tiny LoI}}}}

\newcommand{\LoCI}{\ensuremath{\operatorname{LoCI}}\xspace}
\newcommand{\LOCI}{\LoCI}

\newcommand{\LOCIvariant}[1]{\mathrel{#1_\text{\tiny LoCI}}}

\newcommand{\LOCIleq}{\LOCIvariant{\sqsubseteq}}
\newcommand{\LOCIgeq}{\LOCIvariant{\sqsupseteq}}
\newcommand{\LOCIjoin}{\LOCIvariant{\sqcup}}

\newcommand{\LoCIleq}{\LOCIleq}
\newcommand{\LoCIgeq}{\LOCIgeq}

\newcommand{\sLOCIbigJoin}{\textstyle{\bigsqcup_{\text{\tiny LoCI}}}}

\newcommand{\COCI}{\ensuremath{\mathbf{CoCI}}\xspace}

\newcommand{\poleq}{\sqsubseteq}

\newcommand{\bigjoin}{\bigsqcup}
\newcommand{\biglub}{\bigjoin}

\newcommand{\cell}[1]{{[}#1{]}}

\newcommand{\All}{\rAll}
\newcommand{\Id}{\rId}

\newcommand{\rId}{\mathrm{Id}}
\newcommand{\rAll}{\mathrm{All}}

\newcommand{\ERarrow}{\Rightarrow}

\newcommand{\pref}[2]{\hyperref[#2]{#1\ref*{#2}}}

\newcommand{\Sec}[2][\S]{\pref{#1}{sec:#2}}
\newcommand{\Fig}[2][Fig.]{\pref{#1~}{fig:#2}}

\newcommand{\ie}{i.e.\ }

\newcommand{\consi}[1]{\widetilde{#1}}

\newcommand{\id}[1][]{\operatorname{id}_{#1}}

\newcounter{RuleRef}

\newcommand{\ruleref}[1]{\textsc{\hyperref[rule:#1]{#1}}}

\newif\ifceff
\cefffalse

\newcommand{\domOne}{\textbf{1}} \newcommand{\domTwo}{\textbf{2}} \newcommand{\one}{{*}}

\newcommand{\infer}[3][]{\ensuremath{\displaystyle\frac{\begin{array}{c}#2\end{array}}{\begin{array}{c}#3\end{array}}~\text{#1}}}

\newcommand{\set}[1]{\{#1\}}
\newcommand{\setdef}[2]{\set{#1 \mid #2}}

\newcommand{\tiarrow}{\mathrel{{\Rightarrow}^{\mathsc{ti}}}}

\newcommand{\pker}[1]{\mathrel{\ker_{\poleq}(#1)}}
\newcommand{\qpo}[1]{\mathrel{{\poleq}_{#1}}}
\newcommand{\kset}[1]{K_{#1}}
\newcommand{\okset}[1]{K^{\poleq}_{#1}}

\newcommand{\Cp}{\mathrm{Cp}}
\newcommand{\Er}{\mathrm{Er}}

\newcommand{\EM}{\mathrm{EM}}
\newcommand{\powerdomain}{\powerset}
\newcommand{\plotkinunion}{\mathrel{\bar{\cup}}}

\newcommand{\kleisli}[1]{{{#1}^\dagger}}
\newcommand{\Cv}{\mathrm{Cv}}

\newcommand{\NDchoice}{\mid} 

\begin{document}

\title[Reconciling Shannon and Scott with a Lattice of Computable Information]{Reconciling Shannon and Scott \\ with a Lattice of Computable Information}

\author{Sebastian Hunt}
\email{s.hunt@city.ac.uk}
\orcid{0000-0001-7255-4465}
\affiliation{\institution{City, University of London}
\city{London}
\country{United Kingdom}
}
\author{David Sands}
\email{dave@chalmers.se}
\orcid{0000-0001-6221-0503}
\affiliation{\institution{Chalmers University of Technology}
\city{Gothenburg}
\country{Sweden}
}
\author{Sandro Stucki}
\authornote{This publication was written while the third author was at Chalmers, prior to joining Amazon.}
\email{sastucki@amazon.com}
\orcid{0000-0001-5608-8273}
\affiliation{\institution{Amazon Prime Video}
\city{Gothenburg}
\country{Sweden}
}

\begin{abstract}
  This paper proposes a reconciliation of two different theories of information. The first, originally proposed in a lesser-known work by Claude Shannon (some five years after the publication of his celebrated quantitative theory of communication), describes how the information content of channels can be described \emph{qualitatively}, but still abstractly, in terms of \emph{information elements}, where information elements can be viewed as equivalence relations over the data source domain. Shannon showed that these elements have a partial ordering, expressing when one information element is more informative than another, and that these partially ordered information elements form a complete lattice.  In the context of security and information flow this structure has been independently rediscovered several times, and used as a foundation for understanding and reasoning about information flow.  
  
  The second theory of information is Dana Scott's domain theory, a mathematical framework for giving meaning to programs as continuous functions over a particular topology. Scott's partial ordering also represents when one element is more informative than another, but in the sense of computational progress -- i.e.\ when one element is a more defined or evolved version of another.
  
  To give a satisfactory account of information flow in computer programs it is necessary to consider both theories together, in order to understand not only what information is conveyed by a program
  (viewed as a channel, \`a la Shannon)
  but also how the precision with which that information can be observed is determined by the definedness of its encoding
  (\`a la Scott).
  To this end we show how these theories can be fruitfully combined, by 
  defining \emph{the Lattice of Computable Information} (\LOCI), a lattice of preorders rather than equivalence relations. \LOCI retains the rich lattice structure of Shannon's theory, filters out elements that do not make computational sense, and refines the remaining information elements to reflect how Scott's ordering captures possible varieties in the way that information is presented.
 
 We show how the new theory facilitates the first general definition of termination-insensitive information flow properties, a weakened form of information flow property commonly targeted by static program analyses.

\end{abstract}


\maketitle

\section{Introduction}
\label{sec:intro}

\textbf{Note to Reader:} this paper is not about information theory \cite{Shannon:Mathematical:1948}, but about a theory of information \cite{Shannon:Lattice}.

\subsection{What is the Information in Information Flow?}

The study of information flow is central to understanding many properties of computer programs, and in particular for certain classes of confidentiality and integrity properties.  In this paper we are concerned with providing a better semantic foundation for studying information flow.  

The starting point for understanding information flow is to understand information itself.  Shannon's celebrated theory of information \cite{Shannon:Mathematical:1948}
naturally comes to mind, but Shannon's theory is a theory about \emph{quantities} 
of information, and purposefully abstracts from the information itself.  In a relatively obscure paper\footnote{With around 150 citations, a factor of 1000 fewer than his seminal work on information theory \cite{Shannon:Mathematical:1948} (source: Google Scholar); according to \citet{Rioul+:22}, all but ten of these actually intended to cite the 1948 paper.},
\citet{Shannon:Lattice}
himself notes:
\begin{quotation}
$\ldots$  $H(X)$  [the entropy of a channel $X$] can hardly be said to represent the actual information. Thus, two entirely different sources might produce
information at the same rate (same $H$) but certainly they are not producing the
same information.
\end{quotation}
Shannon goes on to introduce the term \emph{information elements} to denote the information itself.  The concept of an information element can be derived by considering some channel -- a random variable in Shannon's world, but we can think of it as simply a function $f$ from a ``source'' domain to some ``observation'' codomain -- and asking what information does $f$ produce about its input.  Shannon's idea was to view the information itself as the set of functions which are equivalent, up to bijective postprocessing, with $f$, i.e.\ $\setdef{b \circ f}{\text{$b$ is bijective on the range of $f$}}$ -- in other words, all the alternative ways in which the information revealed by $f$ might be faithfully represented. 

Shannon observed that information elements have a natural partial ordering, reflecting when one information element is subsumed by (represents more information than) another, and that any set of information elements relating to a common information source domain can be completed into a \emph{lattice}, with a least-upper-bound representing any information-preserving combination of information elements, and a greatest-lower-bound, representing the common information shared by two information elements, thus providing the title of Shannon's note: ``A Lattice Theory of Information''.  Shannon observes that any such lattice of information over a given source domain can be embedded into a general and well-known lattice, namely the lattice of equivalence relations over that source domain \cite{Ore:Theory}. In fact, the most precise lattice of information for a given domain, i.e.\ the one containing all information elements over that domain, is isomorphic to the lattice of equivalence relations over that domain.  In the remainder of this paper will think in terms of the most precise lattice of information for any given domain. 

This lattice structure, independently dubbed \emph{the lattice of information} by \citet{Landauer:Redmond:CSFW93}, can be used in a uniform way to phrase a large variety of interesting information flow questions, from simple confidentiality questions (is the information in the public output channel of a program no greater than in the public input data?), to arbitrarily fine-grained, potentially conditional information flow policies.  
The lattice of information, described in more detail in \S\ref{sec:loi}, is the starting point of our study. 

\subsection{Shortcomings of the Lattice of Information}
\label{sec:shortcomings}
The lattice of information provides a framework for reasoning about a variety of information flow properties in a uniform way. It is natural in this approach, to view programs as functions from an input domain to some output domain.  But this is where we hit a shortcoming in the lattice of information: program behaviours may be \emph{partial}, ranging from simple nontermination, to various degrees of partiality when modelling structured outputs such as streams.  While these features can be modelled in a functional way using domain theory (see e.g.~\cite{AbramskyJ94book}) the lattice of information is oblivious to the distinction between degrees of partiality. 

Towards an example, consider the following two Haskell functions: 
\begin{lstlisting}[belowskip=5pt]
    parity1 x = if even x then 1 else 0
\end{lstlisting}
\begin{lstlisting}    
    parity2 x = if even x then "Even" else "Odd"
\end{lstlisting}
Even though these two functions have different codomains, intuitively they release the same information about their argument, albeit encoded in different ways.  In Shannon's view they represent the same information element. The information released by a function $f$ can be represented simply by its \emph{kernel} -- the smallest equivalence relation that relates two inputs whenever they get mapped to the same output by $f$.  It is easy to see that the two functions above have the same kernel. 

What about programs with partial behaviours?  A natural approach is to follow the denotational semantics school, and model nontermination as a special ``undefined'' value, $\bot$, and more generally to capture nontermination and partiality via certain families of partially ordered sets (\emph{domains} \cite{AbramskyJ94book}) and to model programs as continuous functions between domains.  
Consider this example: 
\begin{lstlisting}
    parity0 x = if even x then 1 else parity0 x
\end{lstlisting}
Here the program returns \lstinline{1} if the input is even, and fails to terminate otherwise. The kernel of (the denotation of) this function is the same as the examples above, which means that it is considered to reveal the same amount of information.  But intuitively this is clearly not the case: \lstinline!parity0! provides information in a less useful form  than \lstinline!parity1!. When the input is odd, an observer of \lstinline!parity0! will remain in limbo, waiting for an output that never comes, whereas an observer of \lstinline!parity1! will see the value 0 and thus learn the parity of the input. The two are only equivalent from Shannon's perspective if we allow \emph{uncomputable} postprocessors.
(Of course, we are abstracting away entirely from timing considerations here. This is an intrinsic feature of the denotational model, and a common assumption in security reasoning.) 

Intuitively, \lstinline!parity0! provides information which is consistent with \lstinline{parity1}, but the ``quality'' is lower, since some of the information is encoded by nontermination.

Now consider programs \lstinline!A! and \lstinline!B!,
where the input is the value of variable \lstinline{x} and the output domain is a channel on which multiple values may be sent.
Program \lstinline!A! simply outputs the absolute value of \lstinline{x}.
Program \lstinline!B! outputs the same value but in unary, as a sequence of outputs, then silently diverges.

\begin{wrapfigure}[10]{r}{0.3\textwidth}\begin{pseudocode}
A:  output(abs(x))

B:  y := abs(x);
    for i := 1 to y {
      output ()
    };
    while True { };
\end{pseudocode}
\end{wrapfigure}Just as in the previous example, \lstinline{A} and \lstinline{B} compute functions
which have the same kernel, so in the lattice of information
 they are equivalent.
 But consider what we can actually deduce from \lstinline{B} after observing $n$ output events:
 we know that the absolute value of \lstinline{x} is some value $\geq n$,
 but we cannot infer that it is exactly $n$, since we do not know whether there are more outputs yet to come, or if the program is stuck in the final loop. By contrast, as soon as we see the output of A, we know with certainty the absolute value of \lstinline{x}.
In summary, the lattice of information fails to take into account that information can be encoded at different degrees of definedness\footnote{Here we have drawn, albeit very informally, on foundational ideas developed by \citet{Smyth:Powerdomains}, \citet{Abramksy:Thesis,Abramsky:DomainTheoryLogical} and \citet{Vickers:TopologyLogic}, which reveal deep connections between domain theory, topology and logics of observable properties.}.

A second shortcoming addressed in this paper, again related to the lattice of information's unawareness of nontermination, is its inability to express, in a general non ad hoc way, a standard and widely used weakening of information flow properties to the so-called \emph{termination-insensitive} properties \cite{Sabelfeld:Sands:PER,Sabelfeld:Myers:Language}. (When considering programs with stream outputs, they are also referred to as \emph{progress-insensitive} properties \cite{Askarov:Sabelfeld:Tight}.)  These properties are weakenings of information flow policies which ignore any information which is purely conveyed by the definedness of the output (i.e.\ termination in the case of batch computation, and progress in the case of stream-based output).

\subsection{Contributions}

\para{Contribution 1: A refined lattice of information} 
In this paper we present a new abstraction for information, the \emph{Lattice of Computable Information} (\LoCI),  which reconciles Shannon's lattice of information with Scott's domain ordering (\Sec{loci}).  It does so by moving from a lattice of equivalence relations to a lattice of preorder relations, where the equivalence classes of the preorder reflect the ``information elements'', and the ordering between them captures the distinction in quality of information that arises through partiality and nontermination (the Scott ordering).
Just as with the lattice of information, \LoCI induces an information ordering relation on functions; in this ordering, \verb!parity0! is less than \verb!parity1!, but \verb!parity1! is still equivalent to \verb!parity2!.  
Similarly programs $A$ and $B$ above are related but not equivalent. 
We show that \LoCI is, like the lattice of information, well behaved with respect to various composition properties of functions. 

\para{Contribution 2: A generalised definition of termination-insensitive noninterference}
The lattice of computable information gives us the ability to make finer distinctions about information flow with respect to progress and termination.  By modelling this distinction we also have the ability to systematically ignore it; this provides the first uniform generalisation of the definition of termination-insensitive information flow properties (\Sec{tini}).

The remainder of the paper begins with a review of the lattice of information (\Sec{loi}), which is followed by our refinement  (\Sec{loci}), the treatment of termination-insensitivity (\Sec{tini}), 
a discussion of related work (\Sec{related}), and some directions for further work (\Sec{future}).

\section{The Lattice of Information}
\label{sec:loi}

The lattice of information is a way to abstract the information about a data source $D$ which might be revealed by various functions over that data. 
Mathematically, it is simply the set of equivalence relations over $D$, ordered by reverse inclusion, a structure that forms a \emph{complete lattice} \cite{Ore:Theory}, i.e.\ every set of elements in the lattice has a least upper bound and a greatest lower bound.  The lattice of information has been rediscovered in several contexts relating to information and information flow, e.g.\ using partial equivalence relations (PERs) \cite{Hunt:PhD,Hunt:Sands:PER}.  Here we use the terminology from \citet{Landauer:Redmond:CSFW93} who call it the \emph{lattice of information}. 

To introduce the lattice of information let us consider a simple set of values
\[ D = \set{\text{Red}, \text{Orange}, \text{Green}, \text{Blue} } \]
and the following three functions over $D$: \begingroup
\allowdisplaybreaks
\begin{gather*}
\text{isPrimary}(c) = \begin{cases}
  \text{True} & \text{if $c \in \set{\text{Red}, \text{Blue}}$} \\
  \text{False} & \text{otherwise}
\end{cases}
\qquad \quad
\text{isTrafficLight}(c) = \begin{cases}
  \text{False} & \text{if $c = \text{Blue}$} \\
  \text{True} & \text{otherwise}
\end{cases}
\\
\text{primary}(c) = \begin{cases}
  \text{``The primary colour red''} & \text{if $c = \text{Red}$} \\
  \text{``The primary colour blue''} & \text{if $c = \text{Blue}$} \\
  \text{``Not a primary colour''} & \text{otherwise}
\end{cases}
\end{gather*}
\endgroup
Now consider the information that each of these functions reveals about its input:  
\emph{isPrimary} and \emph{isTrafficLight} reveal incomparable information about their inputs -- for example we cannot define either one of them by postprocessing the result of the other. 
The function \emph{primary}, however, not only subsumes both of them, but represents \emph{exactly} the information that the pair of them together reveal about the input, nothing more, nothing less. The lattice of information (over $D$) makes this precise by representing the information itself as an equivalence relation on the elements of $D$.
Elements that are equivalent for a given relation are elements which we can think of as indistinguishable.  
\begin{definition}[Lattice of Information]
For a set $D$, the lattice of information over $D$, $\LOI(D)$, is defined to be the lattice \[
\LOI(D) = ( \ER(D), \LOIleq, \LOIjoin, \LOImeet )
\]
where $\ER(D)$ is the set of all equivalence relations over $D$,
$P \LOIleq Q \eqdef Q \subseteq P$,  the join operation $\LOIjoin$ is set intersection of relations,
    and the meet, $\LOImeet$, is the transitive closure of the set-union of relations. 
\end{definition}
Note that $\LOI(D)$ is a complete lattice (contains all joins and meets, not just the binary ones) \cite{Ore:Theory}. The top element of $\LOI(D)$ is the identity relation on $D$, which we write as $\Id_D$, or just $\Id$ when $D$ is clear from context; the bottom element is the relation which relates every element to every other element, which we write as $\All_D$, or just $\All$. 

In the above definitions we consider equivalence relations to be sets of pairs of elements of $D$.  Another useful way to view equivalence relations is as partitions of $D$ into disjoint blocks (equivalence classes). 
Given an equivalence relation $P$ on a set $D$ and an element $a \in D$, let $[a]_P$ denote the (necessarily unique) equivalence
class of $P$ which contains $a$. Let $[P]$ denote the set of all equivalence classes of $P$. Note that $[P]$ is a partition of $A$. 
\begin{figure}[htb]
    \centering
    \includegraphics[width=0.5\textwidth]{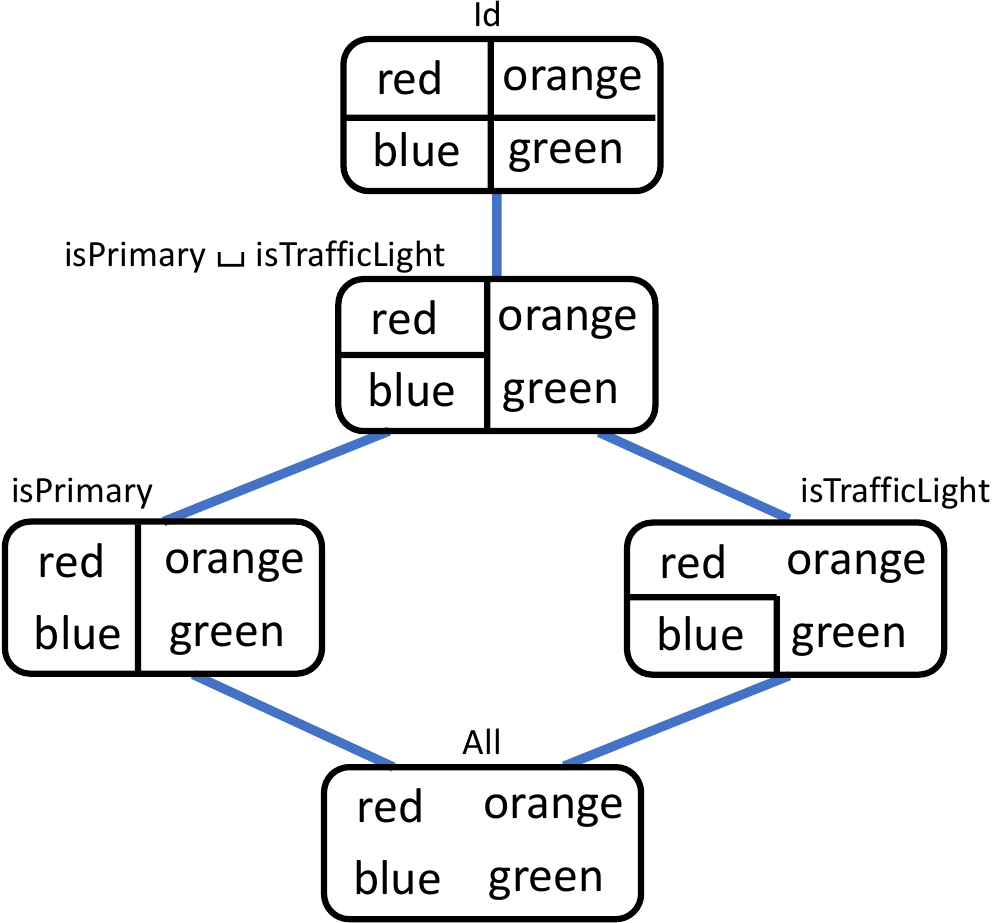}
\caption{An Example Sublattice of the lattice of Information over $\setof{\text{Red}, \text{Orange}, \text{Green}, \text{Blue}}$}
    \label{fig:LOI}
\end{figure}

In \Fig{LOI} we present a Hasse diagram of a sublattice of the lattice of information containing  the points representing the information provided by the functions above, and visualising the equivalence relations by representing them as partitions.  Note that with the partition view, the ordering relation is partition refinement.  The full lattice $\LOI\setof{\text{Red}, \text{Orange}, \text{Green}, \text{Blue}}$ contains 15 elements (known in combinatorics as the \emph{$4^{th}$ Bell number}).

\subsection{The Information Ordering on Functions}
\label{sec:info-order}
To understand the formal connection between the functions and the corresponding information that they release, we use the well-known notion of the \emph{kernel}
of a function: We recall that the \emph{kernel} of a function $f: D \rightarrow E$ is the equivalence relation $\ker(f) \in \LOI(D)$ which relates all elements mapped by $f$ to the same result:
$a \mathrel{\ker(f)} b$ iff $f(a) = f(b)$.

Thus the points illustrated in the lattice do indeed correspond to the respective kernels of the functions, and it can readily be seen that $\ker(\text{primary}) = \ker(\text{isPrimary}) \LOIjoin \ker(\text{isTrafficLight})$.

Note that taking kernels induces an information preorder on any functions $f$ and $g$ which have a common input domain (we write $\dom(f) = \dom(g)$), namely $f \infoleq g \iff \ker(f) \LOIleq \ker(g)$, i.e.\ $g$ reveals at least as much information about its argument as $f$.

Note that this information ordering between functions can be characterised in a number of ways.  
\begin{proposition}
\label{prop:postprocessing}
For any functions $f$ and $g$ such that $\dom(f) = \dom(g)$ the following are equivalent:
\begin{enumerate}
    \item $f \infoleq g$
    \item $\setdef{p \circ f }{ \codom(f) = \dom(p)} \subseteq 
           \setdef{p\circ g }{ \codom(g) = \dom(p)} $ (where $\codom(f)$ is the codomain of function $f$)
    \item There exists $p$ such that $f = p \circ g$
\end{enumerate}
\end{proposition}
The proposition essentially highlights the fact that the information ordering on functions can alternatively be understood in terms of postprocessing (the function $p$). The set $\setdef{p \circ f }{ \codom(f) = \dom(p)}$ can be viewed as all the things which can be computed from the result of applying $f$. 

\subsection{An Epistemic View}
\label{sec:epistemic}
In our refinement of the lattice of information we will lean on an \emph{epistemic} characterisation of the function ordering which focuses on the facts which an observer of the output of a function might learn about its input. 
\begin{definition}\label{def:knowledge-set}
  For $f: A \rightarrow B$ and $a \in A$, define the $f$-\emph{knowledge set} for $a$ as:
\[
\kset{f}(a) = \setdef{ a' \in A }{ f(a) = f(a') }
\]
\end{definition}
The knowledge set for an input $a$ is thus what an observer who knows the function $f$ can maximally deduce about the input if they only get to observe the result, $f(a)$. 
For example, $\kset{\text{primary}}(\text{Green}) = \setdef{c}{\text{primary}(c) = \text{primary}(\text{Green})} =  \set{\text{Green},\text{Orange}}$. 
Note that although we use the terminology ``knowledge'', following work on the semantics of dynamic security policies \cite{Askarov:Sabelfeld:Gradual,Askarov:Chong:Learning}, it is perhaps more correct to think of this as \emph{uncertainty} in the sense that a smaller set corresponds to a more precise deduction.  The point here is that  $f \infoleq g$ can be characterised in terms of knowledge sets: 
$g$ will produce knowledge sets which are at least as precise as those  of $f$:
\begin{proposition}\label{prop:knowledge-set}
Let $f$ and $g$ be any two functions with domain $A$.
Then $f \infoleq g$ iff $\kset{g}(a) \subseteq \kset{f}(a)$ for all $a \in A$.
\end{proposition}

\subsection{Information Flow and Generalised Kernels}
Although we can understand the information released by a function by considering its kernel as an element of the lattice of information, for various reasons it is useful to generalise this idea. The first reason is that we are often interested in understanding the information flow through a function when just a part of the function's output is observed.  For example, if we want to know whether a function is secure, this may require verifying that the public parts of the output reveal information about at most the non-secret inputs.  The second reason to generalise the way we think about information flow of functions is to build compositional reasoning principles.  Suppose that we know that a function $f$ reveals information $P$ about its input. Now suppose that we wish to reason about $f \circ g$.  In order to make use of what we know about $f$ we need to understand the information flow of $g$ when the output is ``observed'' through $P$.  This motivates the following generalised information flow definition (the specific notation here is taken from \cite{Hunt:PhD,Sabelfeld:Sands:PER}, but we state it for arbitrary binary relations \`a la logical relations  \cite{Reynolds:1983}):
\begin{definition}
\label{def:TwoRunsParrowQ}
Let $P$ and $Q$ be binary relations on sets $A$ and $B$, respectively.
Let $f: A \rightarrow B$.

Define:
\[
  f: P \ERarrow Q \iff
  \forall a, a'.
  ({a \mathrel{P} a'} \quad \text{implies} \quad {{f(a)} \mathrel{Q} {f(a')}})
\]
\end{definition}
When $P$ and $Q$ are equivalence relations, these definitions describe information flow properties of $f$ where $P$ describes an upper bound on what can be learned about the input when observing the output ``through'' $Q$ (i.e.\ we cannot distinguish $Q$-related outputs). 

We can read $f: P \ERarrow Q$ as an information flow typing at the semantic level.  As such it can be seen to enjoy natural composition and subtyping properties.  Again, we state these in a more general form as we will reuse them for different kinds of relation:
\begin{fact}
\label{fact:sub-comp}
The following inference rules are valid for all functions and binary relations of appropriate type:\[
\infer[\emph{Sub}]{\;P' \subseteq P \;\;\; f: P \ERarrow Q \;\;\; Q \subseteq Q'\;}
{\;f: {P'} \ERarrow {Q'}\;}
\hspace{3em}
\infer[\emph{Comp}]{\;f: P \ERarrow Q \;\;\; g: Q \ERarrow R\;}
{\;g \circ f: {P} \ERarrow {R}\;}
\]
When these relations are elements of the lattice of information, the conditions
$P' \subseteq P$ and $Q \subseteq Q'$ in the Sub-rule amount to
$P' \LOIgeq P$ and  $Q \LOIgeq Q'$,
respectively.

\end{fact}

Information flow properties also satisfy weakest precondition and strongest postcondition-like properties.
To present these, we start by generalising the notion of kernel of a function:
\begin{definition}[Generalised Kernel]
  Let $\Rel(A)$ denote the set of all binary relations on a set $A$.
  For any $f: A \to B$, define
  $f^\ast: \Rel(B) \to \Rel(A)$ as follows:
  \[ {x \mathrel{f^\ast(R)} y} \iff {f(x) \mathrel{R} f(y)} \]
  We call this the \emph{generalised kernel map}, since $\ker(f) = f^\ast(\Id)$.
\end{definition}
Now, it is evident that $f^\ast$ preserves reflexivity, transitivity and symmetry,
so restricting $f^\ast$ to equivalence relations immediately yields a well defined map in \LoI
(\citet{Landauer:Redmond:CSFW93} use the notation $f\#$ for this map).
Moreover, we can define a partner $f_!$, which operates in the opposite direction and has dual properties (as formalised below):
\begin{definition}
\label{def:wp-sp-LoI}
  For $f : A \to B$:
  \begin{enumerate}
    \item $f^\ast : \LoI(B) \to \LoI(A)$ is the restriction of the generalised kernel map to $\LoI(B)$.
    \item $f_! : \LoI(A) \to \LoI(B)$ is given by
    ${f_!(P)} = {\sLOIbigJoin \setdef{Q \in \LoI}{f : P \ERarrow Q}}$.
\end{enumerate}
\end{definition}
Note that we are overloading our notation here, using $f^\ast$ for both the map on $\REL$ and its restriction to $\LoI$. Later, in \S\ref{sec:infoLoCI}, we overload it again (along with $f_!$).
Our justification for this overloading is that in each case these maps are doing essentially the same thing\footnote{This can be made precise, categorically. See \S\ref{sec:categorical}.}:
$f^\ast(Q)$ is the \emph{weakest precondition} for $Q$ (i.e.
the smallest $P$ such that $f: P \ERarrow Q$)
while $f_!(P)$ is the \emph{strongest postcondition} of $P$
(i.e.\ the largest $Q$ such that $f: P \ERarrow Q$),
where ``smallest'' and ``largest'' are interpreted within the relevant lattice
(\LoI here, our refined lattice \LoCI later).
The following proposition formalises this for $\LoI$
(see Proposition~\ref{prop:wp-sp-LoCI} for its \LoCI counterpart):
\begin{proposition}
\label{prop:wp-sp-LoI}
For any $f: A \to B$, $f^\ast$ and $f_!$ are monotone and,
for any $P \in \LoI(A)$ and $Q \in \LoI(B)$,
the following are all equivalent: 

\hfill 
\emph{(1)}   ${f: P \ERarrow Q}$
\hfill    \emph{(2)}    ${{f^\ast(Q)} \LOIleq P}$
 \hfill   \emph{(3)}    ${Q \LOIleq {f_!(P)}}$
\hfill \mbox{}
\end{proposition}

We have summarised a range of key properties of the lattice of information that make it useful for both formulating a wide variety of information flow properties, as well as proving them in a compositional way. An important goal in refining the lattice of information will be to ensure that we still enjoy properties of the same kind. 
\typeout{End of section ###################################}

\section{LoCI: The Lattice of Computable Information}
\label{sec:loci}

Our goal in this section is to introduce a refinement of noninterference which accounts for the difference in quality of knowledge that arises from nontermination, or more generally partiality, for example when programs produce output streams that may at some point fail to progress. We will assume that a program is modelled in a domain-theoretic denotational style, as a continuous function between partially ordered sets. In this setting, the order relation on a set of values models their relative degrees of ``definedness''.
Simple nontermination is modelled as a bottom element, $\bot$, and in general the ordering relation reflects the evolution of computation. Following Scott, the pioneer of this approach, when one element $d$ is dominated by another $e$, one can think of $e$ as containing \emph{more information} than $d$.  In the domain-theoretic view, a partial element is not a concrete observation or outcome, but a degree of knowledge about a computation.
In this sense $\bot$ represents no knowledge -- you do not fully observe a nonterminating computation, it may still evolve into some more defined result.
Note how this view emphasises how we are abstracting away from time. 
This also explains the basic requirement that all functions (which will be the denotation of programs) are monotone: if you know more about the input (in Scott's sense) you know more about the output.
In domain theory (a standard reference is \cite{AbramskyJ94book}) one restricts attention to
some subclass of well-behaved partially ordered sets (the \emph{domains} of the theory),
in order that recursive computations may be given denotations as least fixpoints.
Being well-behaved in this context entails the existence of suprema of directed sets (and usually a requirement that the domain has a finitary presentation in terms of its compact elements).
In this paper we keep our key definitions as general as possible by stating them for arbitrary partially ordered sets, but still requiring that the functions under study are continuous (preserve directed suprema when they exist).
We expect that some avenues for future work may require additional structure
to be imposed (see \S\ref{sec:future}).

\subsection{Order-Theoretic Preliminaries}

A \emph{partial order} on a set $A$ is a reflexive, transitive and antisymmetric relation on $A$.
A \emph{poset} is a pair $(A, {\poleq}_A)$ where $\poleq$ is a partial order on $A$.
We typically elide the subscript on ${\poleq}_A$ when $A$ is clear from the context.

The \emph{supremum} of a subset $X \subseteq A$, if it exists, is the least upper bound $\biglub X$ with respect to $\poleq_A$.

A set $X \subseteq A$ is \emph{directed} if $X$ is non-empty and, for all $x_1 \in X, x_2 \in X$, there exists $x' \in X$ such that $x_1 \poleq x'$ and $x_2 \poleq x'$.
For posets $A$ and $B$, a function $f: A \rightarrow B$ is \emph{monotone} iff $a \poleq a'$ implies $f(a) \poleq f(a')$.
A function $f:  A \rightarrow B$ is \emph{Scott-continuous} if,
for all $X$ directed in $A$,
whenever $\biglub X$ exists in $A$ then $\biglub f(X)$ exists in $B$ and is equal to $f(\biglub X)$.
From now on we will simply say continuous when we mean Scott-continuous.
Note:
\begin{enumerate}
    \item Continuity implies monotonicity because $a \poleq a'$ implies both that $\{a,a'\}$ is directed and that $\biglub\{a,a'\} = a'$, while $\biglub\{f(a), f(a')\} = f(a')$ implies $f(a) \poleq f(a')$.
    \item Monotonicity in turn implies that, if $X$ is directed in $A$, then $f(X)$ is directed in $B$.
\end{enumerate}
\begin{notation}
In what follows, we write $f \in [A \to B]$ as a shorthand to mean both that $A$ and $B$ are posets and that $f$ is continuous.
\end{notation}

\subsection{Ordered Knowledge Sets}
Our starting point is the epistemic view presented in \S\ref{sec:epistemic}. Recall that we defined the $f$-knowledge set for an input $a$ to be the set $\setdef{a' }{ f(a') = f(a) }$, which is what we learn by observing the output of $f$ when the input is $a$.
However, as discussed in the introduction to this section, in a domain-theoretic setting, observation of a
\emph{partial} output should be understood as provisional: there may be more to come. This requires us to modify the definition of knowledge set accordingly. What we learn about the input when we see a partial output is that the input could be anything which produces that observation,
\emph{or something greater}. Hence:
\begin{definition}\label{def:ordered-knowledge-set}
For $f \in  [A \to B]$ and $a \in A$, define the \emph{ordered $f$-knowledge set} for $a$ as:
\[
\okset{f}(a) = \setdef{ a' \in A }{ f(a) \poleq f(a') }
\]
\end{definition}
Recall (Proposition~\ref{prop:knowledge-set}) that the LoI preorder on functions
has an alternative characterisation in terms of knowledge sets: the kernel of $g$ is a refinement of (i.e.\ more discriminating than) the kernel of $f$ just when each knowledge set of $g$ is a subset of (i.e.\ more precise than) the corresponding knowledge set of $f$.
Unsurprisingly however, if we compare continuous functions based on their \emph{ordered} knowledge sets, the correspondence with LoI is lost.
Consider the examples \texttt{parity0} and \texttt{parity1} from \S\ref{sec:shortcomings}.
We can model these as functions $f_0, f_1 \in  [Z \rightarrow Z_\bot]$,
where $Z$ is discretely ordered (the partial order is just equality)
and the lifting $Z_\bot$ adds a new element $\bot$ which is $\poleq$ everything.
We have:
\begin{center}
    $f_0(x) = \left\{ \begin{array}{cl}
            1       & \mbox{ if } x \mbox{ is even} \\
            \bot    & \mbox{ if } x \mbox{ is odd}
        \end{array}\right.$
        \hspace{3em}and\hspace{3em}
    $f_1(x) = \left\{ \begin{array}{cl}
            1   & \mbox{ if } x \mbox{ is even} \\
            0   & \mbox{ if } x \mbox{ is odd}
        \end{array}\right.$
\end{center}
As discussed previously, these two functions have the same kernel and so are LoI-equivalent.
Moreover, in accordance with Proposition~\ref{prop:knowledge-set},
it is easy to see that they induce the same knowledge sets:
$\kset{f_0}(x) = \kset{f_1}(x)$ for all $x \in Z$.
However, they do not induce the same \emph{ordered} knowledge sets.
In particular, when $x$ is odd we have
$\okset{f_1}(x) = \setdef{ y \in Z }{ y \mbox{ is odd} }$
but
$\okset{f_0}(x) = Z$.
In fact, not only do the two functions induce different ordered knowledge sets, but
$f_1$ is (strictly) more informative than $f_0$, since
$\okset{f_1}(x) \subseteq \okset{f_0}(x)$ for all $x$
(and $\okset{f_1}(x) \neq \okset{f_0}(x)$ for some $x$).

Our key insight is that it is possible to define an alternative information lattice,
one which corresponds exactly with ordered knowledge sets,
by using (a certain class of) preorders, in place of the equivalence relations used in LoI.

\subsection{Ordered Kernels}
A \emph{preorder} is simply a reflexive and transitive binary relation.
Clearly, every equivalence relation is a preorder, but not every preorder is an equivalence relation.
As with equivalence relations, it is possible to present a preorder in an alternative form, as a partition rather than a binary relation, but with one additional piece of information: a partial order on the blocks of the partition. In fact, there is a straightforward 1-1 correspondence between preorders and partially ordered partitions:
\begin{enumerate}
    \item Given a preorder $Q$ on a set $A$, for each $a \in A$, define
    $\cell{a}_Q = \setdef{a' }{ {a \mathrel{Q} a'} \wedge {a' \mathrel{Q} a}}$
    and $\cell{Q} = \setdef{ \cell{a}_Q }{ a \in A }$.
    (Note: although we appear to be overloading the notation introduced in \S\ref{sec:loi}, the definitions agree in the case that $Q$ is an equivalence relation.)
    
    Then define $\cell{a_1}_Q \qpo{Q} \cell{a_2}_Q$ iff $a_1 \mathrel{Q} a_2$.
    This is a well-defined partial order on $\cell{Q}$.
    
    \item Conversely, given a poset
    $(\Phi, \poleq)$,
    where $\Phi$ is a partition of set $A$,
    we recover the corresponding preorder on $A$ by
    defining $a \mathrel{Q} a'$ iff $\cell{a}_\Phi \poleq \cell{a'}_\Phi$.
\end{enumerate}
For a preorder $Q$,
we refer to the equivalence relation with equivalence classes $\cell{Q}$
as the \emph{underlying} equivalence relation of $Q$.
Clearly, the underlying equivalence relation of $Q$ is just $Q \cap Q^{-1}$.

Taking the same path for kernels that we took from unordered to ordered knowledge sets, we arrive at the following definition:
\begin{definition}
Let $(B,\poleq)$ be a poset.
Given $f \in  [A \rightarrow B]$,
define its \emph{ordered kernel} $\pker{f}$ to be $f^\ast(\poleq)$,
thus $x \pker{f} y$ iff $f(x) \poleq f(y)$.
\end{definition}
\begin{proposition}
\label{prop:forgetting}
$\pker{f}$ is a preorder, and its underlying equivalence relation is $\ker(f)$.
\end{proposition}
Only some preorders are the ordered kernels of continuous functions.
For example, if $a \poleq a'$ and $Q$ is the ordered kernel of some continuous $f$, then it must be the case
that $a \mathrel{Q} a'$, since $a \poleq a'$ implies $f(a) \poleq f(a')$.
\begin{definition}[Complete Preorder]
Let $A$ be a poset and let $Q$ be a preorder on $A$.
We say that $Q$ is \emph{complete} iff, whenever $X$ is directed in $A$ and $\biglub X$ exists:
\begin{enumerate}
\item $\forall x \in X. \;x \mathrel{Q} (\biglub X)$
\item $\forall a \in A.\;(\forall x \in X.\;x \mathrel{Q} a) \mathrel{\text{implies}}
{(\biglub X) \mathrel{Q} a}$
\end{enumerate}
Note that part (1) entails that every complete $Q$ contains the domain ordering $(\poleq)$.
\end{definition}

It is perhaps more illuminating to see the definition of completeness for $Q$ presented in terms of its corresponding partially ordered partition:
\begin{lemma}
\label{lemma:continuous-quotient}
    Let $A$ be a poset and let $Q$ be a preorder on $A$.
    Then $Q$ is complete iff, whenever $X$ is directed in $A$ and $\biglub X$ exists in $A$,
    $\biglub{\setdef{\cell{x}_Q }{ x \in X }}$ exists in $(\cell{Q}, \qpo{Q})$ and is equal to $\cell{\biglub X}_Q$.
    
    In other words, $Q$ is complete iff the quotient map $(\lambda a. \cell{a}_Q): A \to (\cell{Q}, \poleq_Q)$ is continuous.
\end{lemma}
To round off this section, we establish that the complete preorders on a poset are just the ordered kernels of all the continuous functions with that domain:
\begin{theorem}
\label{theorem:LoCI-ordered-kernel}
Let $A$ be a poset. Then $Q$ is a complete preorder on $A$ iff there is some
poset $B$ and $f \in  [A \rightarrow B]$ such that $Q = {\pker{f}}$.
\begin{proof}
The implication from left to right is established by Lemma~\ref{lemma:continuous-quotient}.

For the implication right to left, assume $f$ is continuous and let $Q = {\pker{f}}$.
Let $X$ be directed in $A$ such that $\biglub X$ exists. Then:
\begin{enumerate}
    \item Let $x \in X$. Since $f$ is monotone, $f(x) \poleq f(\biglub X)$, thus $x \mathrel{Q} (\biglub X)$.
    \item Let $a \in A$ be such that $x \mathrel{Q} a$ for all $x \in X$.
            Then $f(x) \poleq f(a)$ for all $x \in X$,
            hence $(\biglub f(X)) \poleq f(a)$,
            hence $f(\biglub X) \poleq f(a)$.
            Thus $(\biglub X) \mathrel{Q} a$.
\end{enumerate}
\end{proof}
\end{theorem}

\subsection{LoCI}
We now define the lattice of computable information as a lattice of complete preorders, directly analogous to the definition of LoI as a lattice of equivalence relations. In particular, we can rely on the fact that the complete preorders are closed under intersection:
\begin{lemma}
  Let $\{Q_i\}$ be an arbitrary family of complete preorders.
  Then $\bigcap Q_i$ is a complete preorder.
\end{lemma}
\begin{definition}[Lattice of Computable Information]
For a poset $A$, the lattice of information over $A$, $\LOCI(A)$, is defined to be the lattice \[
\LOCI(A) = ( \PRE(A), \LOCIleq, \LOCIjoin)
\]
where $\PRE(A)$ is the set of all complete preorders on $A$,
$P \LoCIleq Q \eqdef Q \subseteq P$, and  
${\sLOCIbigJoin} \eqdef {\bigcap}$ \end{definition}
Since $\LOCI(A)$ has all joins (not just the binary ones), with the bottom element given by $\All_A = A \times A$, and top element $\poleq_A$, it also has all meets, and hence is  a \emph{complete} lattice. Meets are not used in what follows so we do not dwell on them further here.

 As for \LOI, we can define a preorder on (continuous) functions based on their ordered kernels:
 $f \pinfoleq g$ iff ${\pker{f}} \LoCIleq {\pker{g}}$. As claimed above, this corresponds exactly to an ordering of continuous functions based on their ordered knowledge sets:
\begin{proposition}\label{prop:ordered-knowledge-set}
Let $A$ be a poset and let $f$ and $g$ be any two continuous functions with domain $A$.
Then $f \pinfoleq g$ iff $\okset{g}(a) \subseteq \okset{f}(a)$ for all $a \in A$.
\end{proposition}

\subsection{An Example LoCI} In this section we describe $\LOCI(V)$ for the simple four-point domain $V$ shown in \Fig{lifted-vee}.
A Hasse diagram of the lattice structure is shown in \Fig{preorders-lifted-vee}.
On the right  we enumerate all the complete preorders on $V$, presented as partially ordered partitions. Note that we write $a$ to mean the singleton block $\set{a}$, and $ac\bot$ to mean $\set{a,c,\bot}$, etc.

\begin{figure}[htbp]
     \centering
    \hfill
     \begin{subfigure}[b][][c]{0.12\textwidth}
     \includegraphics[width=\textwidth
]{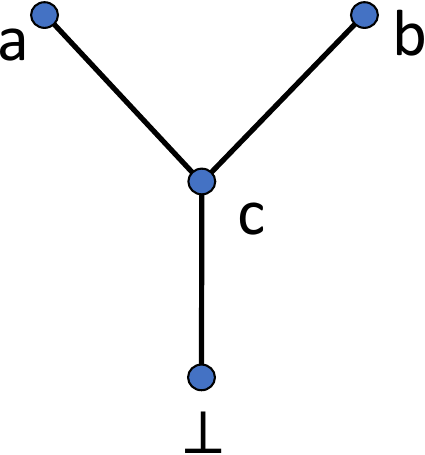}
     \vspace{8ex}
     \caption{Domain $V$}
     \label{fig:lifted-vee}
     \end{subfigure}
     \hfill
\begin{subfigure}[b][][c]{0.75\textwidth}
     \includegraphics[width=\textwidth]{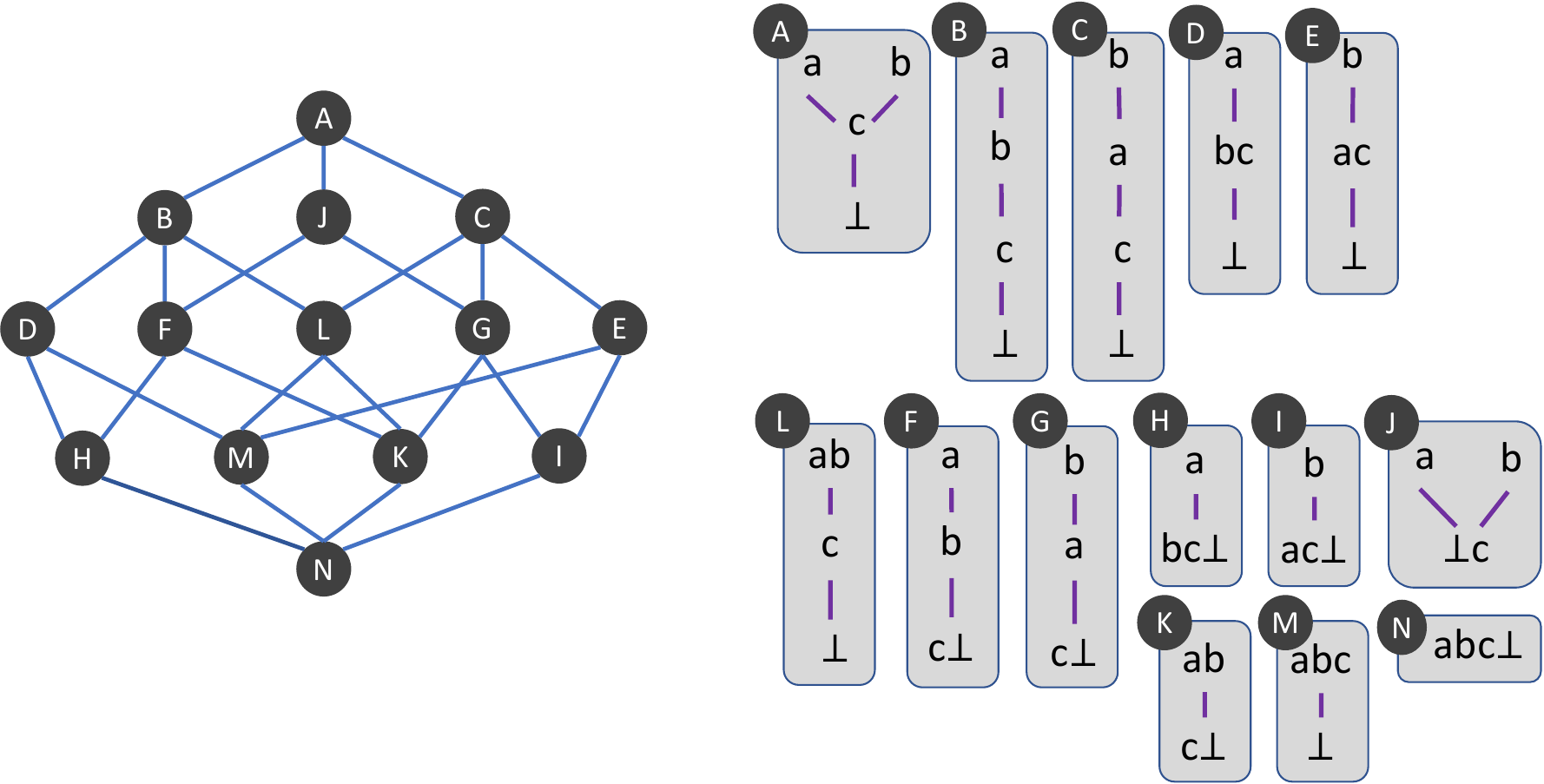}
     \vspace{-5ex}
     \caption{Preorders over $V$}
     \label{fig:preorders-lifted-vee}
 \end{subfigure}
 \vspace{-1ex}
 \label{fig:LOI(V)}
 \caption{$\LoCI(V)$}
\end{figure}

 Let us now consider two continuous functions whose ordered kernels are presented here,
 $f_1, f_2 \in  [V \rightarrow V]$ where:
 \begin{center}
    $f_1 = \lambda x. a$ 
    \hspace{3em}
    $f_2 = \lambda x. \left\{ \begin{array}{cl}
            a & \mbox{ if } x = a \\
            c & \mbox{ if } x \in \{ b, c \} \\
            \bot & \mbox{ if } x = \bot
        \end{array}\right.$ 
 \end{center}
 Since $f_1$ is a constant function it conveys no information about its input, so its ordered kernel is the least element,  \textsf{\small N} ($= \All)$.
 The ordered kernel of $f_2$ is \textsf{\small D}: when the input is $a$, the observer learns this exactly; when the input is $b$ or $c$, the observer learns only that the input belongs to $\{a,b,c\}$; when the input is $\bot$, the observer (inevitably) learns nothing at all.
 Thus, in the \LOCI ordering, $f_2$ is strictly more informative than $f_1$.
 It is interesting to note by contrast, that in the Scott-ordering on functions, $f_1$ is \emph{maximal}, and strictly more defined than $f_2$
 (recall that $f \poleq g$ in the Scott-order iff $f(x) \poleq g(x)$ for all $x$).
 In general, the Scott-ordering between functions tells us little or nothing about their relative capacity to convey information about their inputs. This can be viewed as an instance of the \emph{refinement problem} known from secure information flow \cite{McLean:General}, where a point in a domain can be viewed as its upper set (all its possible ``futures'') and a higher point is then a refinement (a smaller set of futures).

 \subsection{Information Flow Properties in LoCI}
 \label{sec:infoLoCI}
 We can directly use the notation $f : P \ERarrow Q$ introduced earlier to express information flow properties for $P$ and $Q$ in $\LoCI$.  
 Since the ordering on relations in $\LOCI$ is still reversed set containment, both the ``subtyping'' and composition properties stated previously
 (Fact~\ref{fact:sub-comp}) hold equally well for $\LOCI$ as for $\LOI$.
And, as promised, we also have weakest precondition and strongest postcondition properties, provided by appropriate versions of $f^\ast$ and $f_!$ for continuous $f$ and complete preorders:
\begin{definition}
\label{def:wp-sp-LoCI}
  For $f \in [A \to B]$:
  \begin{enumerate}
    \item $f^\ast : \LoCI(B) \to \LoCI(A)$ is the restriction of the generalised kernel map to $\LoCI(B)$.
    \item $f_! : \LoCI(A) \to \LoCI(B)$ is given by
    ${f_!(P)} = {\sLOCIbigJoin \setdef{Q \in \LoCI}{f : P \ERarrow Q}}$.
\end{enumerate}
\end{definition}
(Well-definedness of $f^\ast : \LoCI(B) \to \LoCI(A)$
is slightly less immediate than for the $\LoI$ variant,
but the key requirement is to show that $f^\ast(Q)$ is complete and this follows easily using continuity of $f$.)
The $\LoCI$ analogue of Proposition~\ref{prop:wp-sp-LoI} is then:
\begin{proposition}
\label{prop:wp-sp-LoCI}
For any $f \in [A \to B]$,
$f^\ast$ and $f_!$ are monotone and,
for any $P \in \LoCI(A)$ and $Q \in \LoCI(B)$,
the following are all equivalent:

\hfill
\emph{(1)} ${f: P \ERarrow Q}$
\hfill    
\emph{(2)}     ${{f^\ast(Q)} \LOCIleq P}$
\hfill
\emph{(3)}    ${Q \LOCIleq {f_!(P)}}$
\hfill   
    \mbox{}
\end{proposition}

\newcommand{\PC}{\mathbf{PC}}
\subsection{A Category of Computable Information}
\label{sec:categorical}
Some of the definitions and properties introduced earlier can be recast in category-theoretic terms through the framework of \emph{Grothendieck fibrations}.
In this subsection, we briefly sketch the relevant connections.
The subsection is intended as an outline for interested readers rather than a definitive category-theoretic treatment of $\LOCI$ -- which is beyond the scope of this paper.
The remainder of the paper does not depend on any of the ideas discussed in this subsection, but some notational choices and technical developments are inspired by it.

So far we have treated posets
$A$, $B$ 
and continuous functions 
$f: A \to B$
as a 
semantic framework,
in which we have studied, separately, the 
information associated with individual domains
$A$ via $\LOCI(A)$, 
and the flow of information over a channel $f$ via $f: -\ERarrow-$.
An alternative approach is to combine the information represented by a preorder $P \in \LOCI(A)$ and its underlying poset $A$ into a single mathematical structure, and to study the overall properties of such \emph{information domains}.

\begin{definition}
  An \emph{information domain} is a pair $(A, P)$ consisting of a poset $A$ and a complete preorder $P \in \LOCI(A)$.  An \emph{information-sensitive function} between information domains $(A, P)$ and $(B, Q)$ is a continuous function $f: A \to B$, such that $f: P \ERarrow Q$.
\end{definition}

Information domains and information-sensitive functions form the \emph{category of computable information} $\COCI$.
Identities and composition are defined via the underlying continuous maps; composition preserves information-sensitivity by Fact~\ref{fact:sub-comp}~(Comp).

The category $\COCI$ and the family of lattices $\LOCI(A)$ are related by a \emph{fibration} or,
to use the terminology coined by
\citet{MelliesZ15popl}, by a \emph{type refinement system}.
Intuitively, we may think of a poset $A$ as a \emph{type}, and of an information domain $(A, P)$ as a \emph{refinement} of $A$.
For each type $A$, there is a subcategory of $\COCI$, called the \emph{fibre over $A$}, whose objects are the refinements of $A$, and which is equivalent to $\LOCI(A)$.

Formally, there is a forgetful functor $U$ from $\COCI$ to the category $\PC$ of posets and continuous functions that maps refinements to their types $U(A, P) = A$ and information-sensitive functions to the underlying continuous maps $U(f) = f$.
The fibre $\COCI_A$ over $A$ is the ``inverse image'' of $A$ under $U$, \ie the subcategory of $\COCI$ with objects of the form $(A, P)$ and arrows of the form $\operatorname{id}_A : (A, P) \to (A, Q)$, where $P, Q \in \LOCI(A)$.
Note that the objects of $\COCI_A$ are uniquely determined by their second component, and that there is an arrow between the pair of objects $(A, P)$ and $(A, Q)$ iff $P \LOCIgeq Q$.
In other words, the category $\COCI_A$ is equivalent to the dual lattice of $\LOCI(A)$, thought of as a complete (and co-complete) posetal category.  In line with the terminology of \citeauthor{MelliesZ15popl}, we may call $\COCI_A$ the \emph{subtyping} lattice over $A$.

Furthermore, the functor $U$ is a \emph{bifibration}.
Intuitively, this ensures that we can reindex refinements along continuous maps.
The formal definition of a bifibration is somewhat involved \citep[see e.g.][]{MelliesZ15popl}, but it can be shown, in our setting, to correspond to the existence of weakest preconditions and strongest postconditions as characterised in Proposition~\ref{prop:wp-sp-LoCI}, plus the following identities
\begin{align*}
    \operatorname{id}_A^\ast    &= \operatorname{id}_{\LoCI(A)} & 
    (g \circ f)^\ast            &= f^\ast \circ g^\ast &
    (\operatorname{id}_A)_!     &= \operatorname{id}_{\LoCI(A)} & 
    (g \circ f)_!               &= g_! \circ f_!
\end{align*}
which are easy to prove.
For the last one, rather than showing
${(g \circ f)_!} = {g_! \circ f_!}$ directly -- which is awkward --
it is simpler to show
${(g \circ f)^\ast} = {f^\ast \circ g^\ast}$ first,
and then use the fact that
each $(h^\ast,h_!)$
is an adjoint pair. The \emph{cartesian} and \emph{opcartesian liftings} of $f: A \to B$ to $(B, Q)$ and $(A, P)$ are then given by $f: (A, f^\ast(Q)) \to (B, Q)$ and $f: (A, P) \to (B, f_!(P))$, respectively.

Using the reindexing maps $f^\ast$ and $f_!$, we can extend the poset-indexed set $\{ \COCI_A \}_{A \in \PC}$ of fibres over $A$ into a \emph{poset-indexed category}, that is, a contravariant functor $F: \PC^{\operatorname{op}} \to \mathbf{Cat}$, that maps posets $A$ to fibres $F(A) = \COCI_A$ and whose action on continuous maps $f: A \to B$ is given by
\begin{align*}
    F(f) &: \COCI_B \to \COCI_A \\
    F(f)&(Q) = f^\ast(Q)
\end{align*}
Replacing the reindexing map $f^\ast$ with $f_!$, we obtain a similar, covariant functor $G: \PC \to \mathbf{Cat}$.\footnote{The existence of $F$ and $G$ is in fact sufficient to establish that $U$ is a bifibration.}

The family of lattices $\LOCI(A)$ and the category $\COCI$ fully determine each other: we may obtain $\LOCI(A)$ as the fibres of $\COCI$ via $U$, and conversely, we may reconstruct the category $\COCI$ from 
the indexed category $F$ via the Grothendieck construction $\COCI = \int F$.

Finally, note that the above can also be adapted to the simpler setting of $\LOI$.
In that case, types are simply sets, and refinements are \emph{setoids}, \ie pairs $(S, R)$ consisting of a set $S$ and an equivalence relation
$R \in \LoI(S)$.
The relevant fibration is the obvious forgetful functor $U: \mathbf{Setoid} \to \mathbf{Set}$ from the category of setoids and equivalence-preserving maps to the underlying sets and total functions.

 \subsection{A Partial Embedding of LoI into LoCI}
 As discussed earlier, a key advantage of $\LOCI$ in comparison to $\LOI$ is that it distinguishes between functions which have the same (unordered) kernel but which differ fundamentally in what information they actually make available to an output observer, due to different degrees of partiality.

 But there is another advantage of $\LOCI$: it excludes ``uncomputable'' kernels, those equivalence relations in $\LOI$ which are not the kernel of \emph{any} continuous function.
  Consider the example of $\LOCI(V)$ in \Fig{preorders-lifted-vee}.
  Since $V$ has four elements, there are 15 distinct equivalence relations in $\LOI(V)$. Note, however, that $\LOCI(V)$ has only 14 elements. Clearly then there must be at least one equivalence relation which is being excluded by $\LOCI(V)$
  (in fact, five elements of $\LOI(V)$ are excluded).
 Let us settle on some terminology for this:
 \begin{definition}
   Let $A$ be a poset. Let $R$ be an equivalence relation on $A$ and let $Q$ be a complete preorder on $A$. Say that $Q$ \emph{realises} $R$
if $R$ is the underlying equivalence relation of $Q$.
   When such $Q$ exists for a given $R$, we say that $R$ is \emph{realisable}.
 \end{definition}
 Note that, by Proposition~\ref{prop:forgetting},
 the underlying equivalence relation of ${\pker{f}}$ is ${\ker(f)}$,
 so by Theorem~\ref{theorem:LoCI-ordered-kernel} it is equivalent to say that $R$ is realisable iff $R$ is the kernel of some continuous function.
 
 In $\LOCI(V)$, note that A, B and C all realise the identity relation.
 Similarly, F, G and J all realise the same equivalence relation as each other.
Thus, while $\LOCI(V)$ has 14 elements, together they realise only 10 of the 15 possible equivalence relations over $V$.
 As an example of a missing equivalence relation, consider the one with equivalence classes $\{a, b, \bot\}$, $\{c\}$.
 Recall that a subset $X$ of a poset is \emph{convex} iff,
whenever $x \poleq y \poleq z$ and $x, z \in X$,
then $y \in X$.
 Note that $\{a, b, \bot\}$ is not convex, but it is easy to see that all equivalence classes in the kernel of a monotone function \emph{must} be convex.
 (The convexity test also fails for the four other missing equivalence relations. But convexity alone is not sufficient for realisability, even in the finite case.
 See \S\ref{sec:verifying-realisability} below.)
 
 When an equivalence relation $S$ \emph{is} realisable, we can show that there must be a \emph{greatest} element of LoCI which realises it.
 Moreover, we can use this realiser to re-express an \LoI property
 $f: R \ERarrow S$ as an equivalent \LoCI property. 
 To this end, we define a pair of monotone maps which allow us to move back and forth between $\LoI$ and $\LoCI$:
 \begin{definition}
 For poset $A$ define
 $\Cp_A: \LoI(A) \to \LoCI(A)$ and $\Er_A: \LoCI(A) \to \LoI(A)$ by:
 \begin{enumerate}
     \item $\Cp_A(R)
            = \sLOCIbigJoin \setdef{ P \in \LoCI(A) }{ P \supseteq R }
            = \bigcap \setdef{ P \in \LoCI(A) }{ P \supseteq R }$
     \item $\Er_A(P)$ is the underlying equivalence relation of $P$:
     $\Er_A(P) = P \cap {P^{-1}}$
\end{enumerate}
 It is easy to see that both maps are monotone.
 We will routinely omit the subscripts on $\Cp$ and $\Er$ in contexts where the intended domain is clear.
 \end{definition}
 Note that, by definition, $R \in \LoI(A)$ is realisable iff there exists some $P \in \LoCI(A)$ such that $\Er(P) = R$.
 Now, $\Cp(R)$ is defined above to be the greatest $P \in \LoCI(A)$ such that $P \supseteq R$.
 But observe that
 $\Er(P) \LOIleq R$ iff $\Er(P) \supseteq R$,
and
$\Er(P) = P \cap {P^{-1}} \supseteq R$ iff $P \supseteq R$,
since $R$ is symmetric.
So we have actually defined $\Cp(R)$ to be the greatest $P \in \LoCI(A)$ such that
$\Er(P) \LOIleq R$.
 The following propositions are immediate consequences:
\begin{proposition}
 \label{prop:greatest-realiser}
 $R$ is realisable iff \,$\Er(\Cp(R)) = R$
 (in which case $\Cp(R)$ is its greatest realiser).
\end{proposition}
 \begin{proposition}
 \label{prop:GC}
 The pair $(\Er_A, \Cp_A)$ forms a \emph{Galois connection} between
 \LoCI(A) and \LoI(A). That is to say
 for every $P \in \LoCI(A)$ and every $R \in \LoI(A)$:
\begin{equation}
\label{eq:GC}
    {\Er(P) \LOIleq R} \iff {P \LOCIleq \Cp(R)} \tag{GC}
\end{equation}
 (See \cite{GaloisConnectionsPrimer} for an introduction to Galois connections.)
\end{proposition}
This extends to an encoding of \LoI properties as \LoCI properties:
\begin{theorem}
For all $f \in [A \to B]$, for all $R \in LoI(A)$, for all $Q \in \LoCI(B)$:
\[ f: R \ERarrow \Er(Q) \iff f: \Cp(R) \ERarrow Q
\]
\begin{proof}
By Propositions \ref{prop:wp-sp-LoI} and \ref{prop:wp-sp-LoCI}, it suffices to show
$f^\ast(\Er(Q)) \LOIleq R$ iff $f^\ast(Q) \LoCIleq \Cp(R)$.
First we note that the following holds by an easy unwinding of the definitions:
\begin{equation}
    \label{eqn:f-star-Er-swap}
    \tag{$\ast$}
    f^\ast(\Er_B(Q)) = \Er_A(f^\ast(Q))
\end{equation}
Then we have:
\[
f^\ast(\Er(Q)) \LOIleq R
\iff \Er(f^\ast(Q)) \LOIleq R
\iff f^\ast(Q) \LoCIleq \Cp(R)
\]
where the first equivalence holds by (\ref{eqn:f-star-Er-swap})
and the second by (\ref{eq:GC}).
\end{proof}
\end{theorem}
\begin{corollary}
\label{coroll:encodingLoIinLoCI}
If $S$ is realisable then
$f: R \ERarrow S \iff f: \Cp(R) \ERarrow \Cp(S)$.
\begin{proof}
By Proposition~\ref{prop:greatest-realiser},
$S$ is realisable iff $S = \Er(\Cp(S))$,
so let $Q = \Cp(S)$ in the theorem.
\end{proof}
\end{corollary}
It is interesting to note that Corollary~\ref{coroll:encodingLoIinLoCI} does not require $R$ to be realisable.
However, in general, the equivalence does not hold unless $S$ is realisable.
For a counterexample, consider the three-point lattice $A = \{0,1,2\}$ with
$0 \sqsubset 1 \sqsubset 2$, and let $S$ be the equivalence relation with equivalence classes
$\{0,2\}$ and $\{1\}$.
The first of these classes is not convex, so $S$ is not realisable.
Now consider the property $f: \All \ERarrow S$.
It is easy to see that this property fails to hold for some choices of $f \in [A \to A]$ (e.g. choose $f$ to be the identity).
However, $\Cp(S) = \Cp(\All) = \All$ and $f: \All \ERarrow \All$ holds trivially.

 \subsubsection{Verifying Realisability}
\begin{wrapfigure}{r}{0.3\textwidth}
     \centering
\includegraphics[width=0.25\textwidth,trim=0pt 0pt 0pt 60pt]{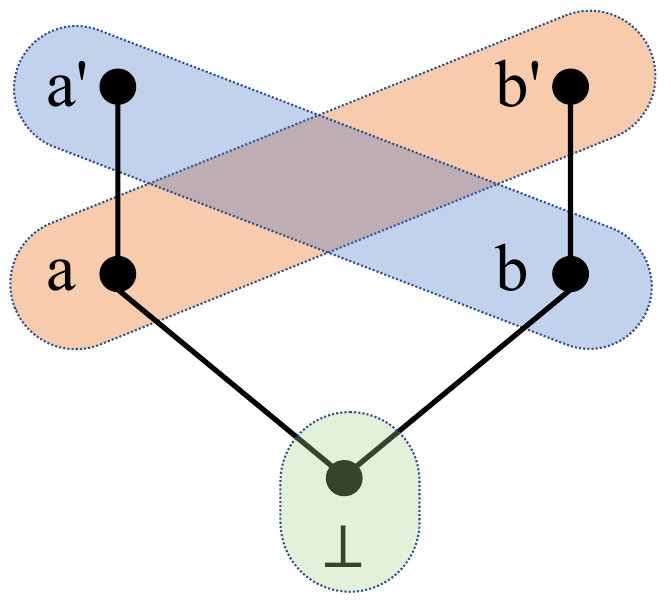}
     \caption{An Unrealisable Equivalence Relation}
     \label{fig:unrealisable-example}
\end{wrapfigure}     
We describe a simple necessary condition for realisability,
 which is also sufficient in the finite case.
 It is motivated by the following example.
 \label{sec:verifying-realisability}
  Let $R$ be the equivalence relation shown in \Fig{unrealisable-example}.
  The three equivalence classes are clearly convex, but $R$ is not realisable. To see why, suppose that $R$ is the kernel of $f$.
There must be distinct elements $x$ and $y$ such that
 $f(\{a,b'\}) = \{x\}$
 and
 $f(\{b,a'\}) = \{y\}$.
 If $f$ is monotone then, since $a \poleq a'$ and $b \poleq b'$, it must be the case that
 $x \poleq y \poleq x$,
 which contradicts the assumption that $x$ and $y$ are distinct.
 
 The example of \Fig{unrealisable-example} generalises quite directly.
 Given any equivalence relation $R$ on a poset $A$,
 define $\phi$ as the relation on $\cell{R}$ which relates two equivalence classes whenever they contain $(\poleq)$-related elements:
 \(
    {\cell{a}_R \mathrel{\phi} \cell{b}_R}
    \iff
    \exists x \in \cell{a}_R. \exists y \in \cell{b}_R. x \poleq y
 \).
 In \Fig{unrealisable-example}, unrealisability manifests as a non-trivial cycle in the graph of $\phi$, that is, a sequence
 \( \cell{a_1}_R \mathrel{\phi} \cdots \mathrel{\phi} \cell{a_n}_R \mathrel{\phi} \cell{a_1}_R\)
 with $n > 1$
 and such that all $\cell{a_i}_R$ are distinct.
 By the obvious inductive generalisation of the above argument,
 any monotone $f$ necessarily maps all $a_i$ to the same value, thus
 making $\ker(f) = R$ impossible.
 So if the graph of $\phi$ contains a non-trivial cycle, $R$ is not realisable.
 (Note also that this generalises the convexity condition:
 if any $\cell{a}_R$ is non-convex,
 there will be a non-trivial cycle with $n = 2$.)
 
 Conversely, to say that $\phi$ is free of such cycles is just to say that the transitive closure $\phi^{+}$ is antisymmetric. Clearly, $\phi^{+}$ is
 also reflexive and transitive, thus
 $B = (\cell{R}, \phi^{+})$ is a poset.
 Let $f: A \to B$ be the map $a \mapsto \cell{a}_R$.
 Then $f$ is monotone (because
 ${x \poleq y}$ implies ${\cell{x}_R \mathrel{\phi} \cell{y}_R}$)
 and
 $\ker(f) = R$.
 In the case that $\poleq_A$ is of finite height, this establishes that $R$ is realisable.
 
 \subsection{Post Processing}
 In \S\ref{sec:loi} we introduced three equivalent ways of ordering functions,
 the first based on inclusion of their kernels ($\infoleq$),
 the second in terms of their inter-definability via postprocessing (Proposition~\ref{prop:postprocessing}),
 and the third in terms of their knowledge sets (Proposition~\ref{prop:knowledge-set}).
 Moving to a setting of posets and continuous functions,
 we have presented direct analogues of the first of these in terms of ordered kernels ($\pinfoleq$),
 and of the third in terms of ordered knowledge sets (Proposition~\ref{prop:ordered-knowledge-set}).
 However, it turns out that there is no direct analogue of the postprocessing correspondence.
 To see why, we consider two pairs of counterexamples which illustrate two essentially different ways in which the postprocessing correspondence fails for $\LOCI$.

\para{Counterexample 1: Non-Existence of a Monotone Postprocessor}
Consider a test \texttt{isEven1} on natural numbers which simply returns True or False.
This can be modelled in the obvious way by a function
$\mathrm{isEven1} \in  [N \rightarrow \mathrm{Bool}_\bot]$,
where $N$ is the unordered set of natural numbers and $\mathrm{Bool}_\bot$ is the lifted domain of Booleans in \Fig{counter-codomains}.
\begin{figure}[bth]
     \centering
     \begin{subfigure}[b]{.4\textwidth}
     \centering
     \includegraphics[width=\linewidth]{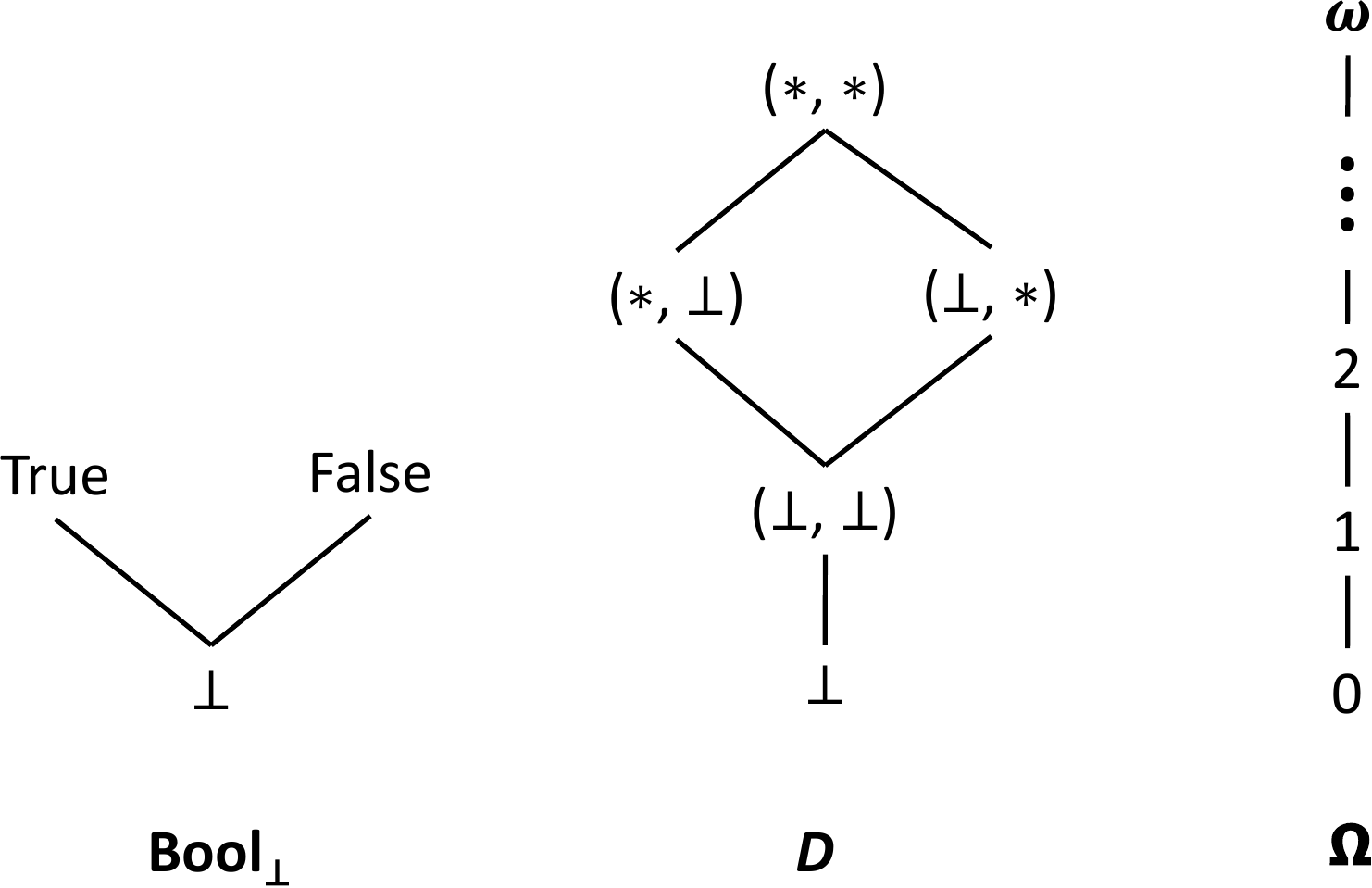}
     \caption{Codomains}
     \label{fig:counter-codomains}
     \end{subfigure}\hspace{3em}
     \begin{subfigure}[b]{.32\textwidth}
     \centering
     \includegraphics[width=0.7\textwidth]{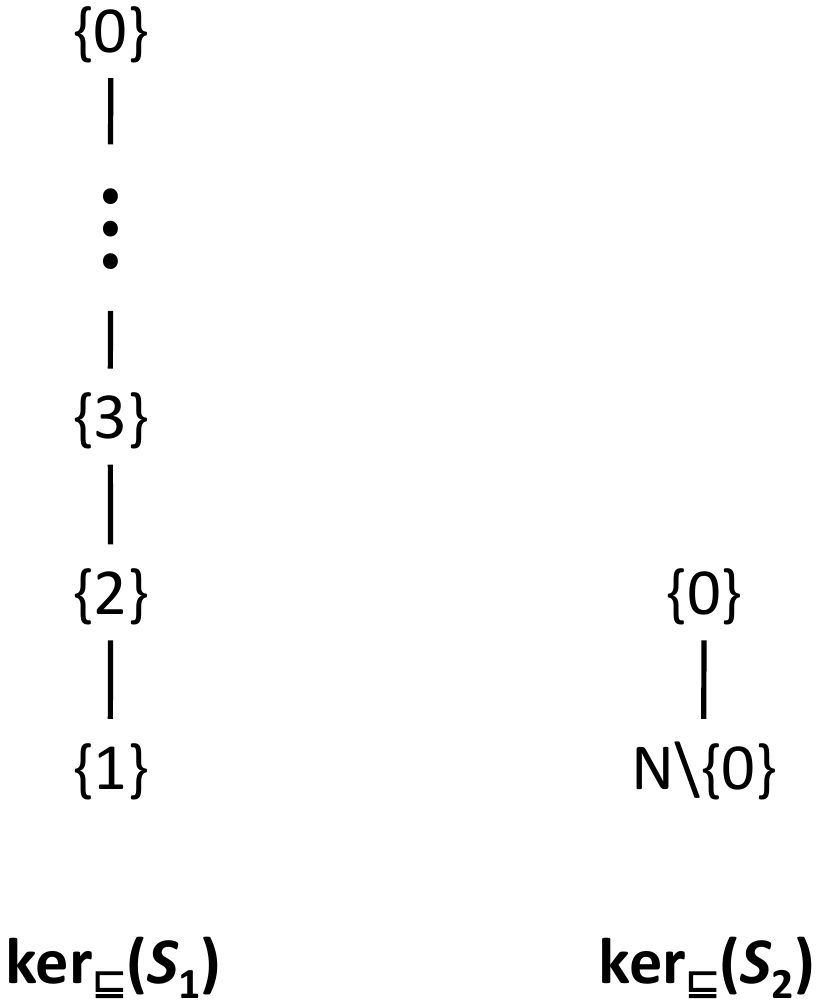}
     \caption{Kernels}
     \label{fig:counter-kernels}
     \end{subfigure}
     \caption{Postprocessing Counterexamples}
 \end{figure}

Now consider the following Haskell-style function definition 
\begin{lstlisting}
    isEven x  = if even x then ((), spin) else (spin, ())
      where spin = spin
\end{lstlisting}

Tuples in Haskell are both lazy and lifted, so this can be modelled by a function
$\mathrm{isEven2} \in  [N \rightarrow D]$, where $D$ is the lifted diamond domain in
\Fig{counter-codomains}.  (Haskell does not have a primitive type for natural numbers, but only integers, but for the sake of the example let us assume that the program operates over naturals.) 

Both these functions have the same kernel (ordered and unordered): it simply partitions $N$ into the sets of even and odd numbers.
So $(\mathrm{isEven2} \pinfoleq \mathrm{isEven1})$ and
$(\mathrm{isEven1} \pinfoleq \mathrm{isEven2})$.
We can certainly obtain isEven2 from isEven1 by postprocessing:
map $\bot$ to $\bot$, map $T$ to $(\ast,\bot)$, and map $F$ to $(\bot,\ast)$.
However, there is no continuous postprocessor $p \in  [D \rightarrow \mathrm{Bool}_\bot]$ such that
$\mathrm{isEven1} = p \circ \mathrm{isEven2}$.
The problem is that any such $p$ must map $(\ast,\bot)$ to $T$ and $(\bot,\ast)$ to $F$.
But then,
since $(\ast,\ast)$ is greater than both $(\ast,\bot)$ and $(\bot,\ast)$,
$p$ must map $(\ast,\ast)$ to a value greater than both $T$ and $F$, and no such value exists.
Note, however that $(\ast,\ast)$ is not actually in the \emph{range} of isEven2.
If $p$ was not required to be monotone, the problem would therefore be easily resolved, since $p$ could arbitrarily map $(\ast,\ast)$ to either $T$ or $F$ (or even to $\bot$).
Unfortunately, such $p$ would not actually be computable.
Nonetheless, it is clear that it is indeed \emph{computationally feasible} to learn exactly the same information from the output of the two functions. For example, we may poll the two elements in the output of isEven2 in alternation, until one becomes defined; as soon as this happens we will know the parity of the input.
This behaviour is clearly implementable in principle, even though it does not define a monotone function in $D \rightarrow \mathrm{Bool}_\bot$.
(Of course, we cannot implement this behaviour in sequential Haskell, but this is just a limitation of the language.)

Conceivably, a slightly more liberal postprocessing condition could be designed to accommodate this and similar counterexamples (allow postprocessors to be partial, for example).

\para{Counterexample 2: Non-Existence of a Continuous Postprocessor}
Consider these two programs:

\begin{minipage}[t]{0.48\textwidth}
\begin{pseudocode}
S1: if (x == 0) while True output();
    for i := 1 to x - 1 {
      output ()
    };
    while True { }
\end{pseudocode}
\end{minipage}\hfill
\begin{minipage}[t]{0.48\textwidth}
\begin{pseudocode}
S2: if (x == 0) while True output();
    while True { }
\end{pseudocode}
\end{minipage}

Both programs take a natural number $x$ and produce a partial or infinite stream of units. They can be modelled by functions $S_1, S_2 \in  [N \rightarrow \Omega]$,
where $\Omega$ is the poset illustrated in \Fig{counter-codomains}. (In the picture for $\Omega$ we represent each partial stream of units by its length; the limit point $\omega$ represents the infinite stream.)
When $x = 0$, both programs produce an infinite stream. When $x > 0$, S1 produces a stream of length $x-1$, and then diverges; S2 simply diverges immediately.

As illustrated in \Fig{counter-kernels}, the ordered kernel for $S_1$ is isomorphic to $\Omega$, while the ordered kernel for $S_2$ is a two-point lattice.
Clearly, $S_2 \pinfoleq S_1$.
But there is no continuous $p \in  [\Omega \rightarrow \Omega]$ such that $S_2 = {p \circ S_1}$. The problem in this case is that $p$ would have to send all the finite elements of $\Omega$ to the bottom point $0$, while sending the limit point $\omega$ to a different value.

The key thing to note here is that, although $\cell{{\pker{S_1}}}$ contains $\{0\}$ as a maximal element
(it is the inverse image under S1 of the infinite output stream)
an observer of S1 will never actually learn that $x = 0$ in finite time.
With each observed output event, the observer rules out one more possible value for $x$,
but there will always be infinitely many possible values remaining.
After observing $n$ output events, the observer knows only that ${x = 0} \vee {x > n}$.
By contrast, an observer of S2 learns that $x = 0$ as soon as the first output event is observed. (On the other hand, when $x > 0$, an S2 observer learns nothing at all.)

Perhaps the best we can claim is that the \LOCI model is conservative, in the sense that it faithfully captures what an observer will learn ``in the limit''. But, as S1 illustrates, sometimes the limit never comes.

\section{Termination-Insensitive Properties}
\label{sec:tini}
In this section we turn to the question of how \(\LoCI\) can help us to formulate the first general definition of a class of weakened information-flow properties known as \emph{termination-insensitive} properties (or sometimes, \emph{progress-insensitive} properties). 

\subsection{What is Termination-Insensitivity?}
We quote \citet{Askarov+:Termination}:
\begin{quotation}
Current tools for analysing information flow in programs build upon
ideas going back to Denning's work from the 70's $\langle$\cite{Denning:Certification}$\rangle$. These systems enforce an
imperfect notion of information flow which has become known as termination-insensitive noninterference. Under this version of noninterference, information
leaks are permitted if they are transmitted purely by the program’s termination behaviour (i.e.\ whether it terminates or not). This imperfection is the price to pay
for having a security condition which is relatively liberal (e.g.\ allowing while-loops whose termination may depend on the value of a secret) and easy to check.
\end{quotation}
The term \emph{noninterference} in the language-based security literature refers to a class of information flow properties built around a lattice of security labels (otherwise known as \emph{security clearance levels}) \cite{Denning:Lattice}, in the simplest case two labels, $H$ (the label for secrets) and $L$ (the label for non-secrets), together with a ``may flow'' partial order $\prec$, where in the simple case $L \prec H$, expressing that public data may flow to (be combined with) secrets.

On the semantic side, for each label $k$ there is a notion of \emph{indistinguishability} between inputs and, respectively, outputs --  equivalence relations which determines whether an observer at level $k$ can see the difference between two different elements.  These relations must agree with the flow relation in the sense that whenever $j\prec k$ then indistinguishability at level $k$ implies indistinguishably at level $j$. Indistinguishability relations are either given directly, or can be constructed as the kernel of some projection function which extracts the data of classification at most $k$. Thus ``ideal'' noninterference for a program denotation $f$ can be stated in terms of the lattice of information as a conjunction of properties of the form $f: P_k \ERarrow Q_k$, expressing that an output observer at level $k$ learns no more than the level-$k$ input.

Without focusing on security policies in particular, we will show how to take any property of the form 
$f: P \ERarrow Q$ and weaken it to a property which allows for termination leaks. The key to this is to use the preorder refinement of $Q$ to get a handle on exactly what leaks to allow.   The case when $P$ and $Q$ are used to model security levels will just be a specific instantiation.  
But even for this instantiation we present a new generalisation of the notion of termination sensitivity beyond the two special cases that have been studied in the literature, namely (i) the ``batch-job'' case when programs either terminate or deliver a result, and (ii) the case when programs output a stream of values. 
In the recent literature the term \emph{progress-insensitivity} has been used to describe the latter case, but in this section we will not distinguish these concepts -- they are equally problematic for a Denning-style program analysis. 
Case (i) we will refer to henceforth as \emph{simple termination-insensitive noninterference} and is relevant when the result domain of a computation is a flat domain.

As a simple example of case (i) consider the programs
\begin{center}
    $A \; = \; {}$\lstinline!while (h>0) { }!
    \qquad and \qquad 
    $B \; = \; {}$\lstinline!while (h>0) {h := h-1}!.
\end{center}
Assume that $h$ is a secret.  Standard information flow analyses notice that the loop condition in each case references variable $h$, but since typical analyses do not have the ability to analyse termination properties of loops, they must conservatively assume that information about $h$ leaks in both cases (when in fact it only leaks for program $A$). This prevents us from verifying the security of any loops depending on secrets. However, a termination-insensitive analysis ignores leaks through termination behaviour and thus both $A$ and $B$ are permitted by termination-insensitive noninterference: such an analysis is more permissive because it allows loops depending on secrets (such as $B$), but less secure because it also allows leaky program $A$ (which terminates only when $h \leq 0$).

Case (ii), progress-insensitivity, is the same issue but for programs producing streams. Consider here two programs which never terminate (thanks to $D\,= \,$ \lstinline!while True { }!):
\begin{center}
    $A' \, = \, {}$ \lstinline!output(1); !$A$\lstinline!; output(1); !$D$
    \qquad versus \qquad
    $B' \, = \, {}$ \lstinline!output(1); !$B$\lstinline!; output(1); !$D$.
\end{center}
 Here $B'$ is noninterfering but $A'$ is not, but both are permitted by the termination-insensitive condition (aka progress-insensitivity) for stream output defined in e.g.~\cite{Askarov+:Termination}.
The point of this  example is to illustrate that the carrier of the information leak is not just the simple ``does it terminate or not'', but the \emph{cause} of the leak is the same.

The definition in \cite{Askarov+:Termination} is ad hoc in that it is specific to the particular model of computation. If the computation model is changed (for example, if there are parallel output streams, or if there is a value delivered on termination) then the definition has to be rebuilt from scratch, and there is no general recipe to do this. 

\subsection{Detour: Termination-Insensitivity in the Lattice of Information} Before we get to our definition, it is worth considering how termination-insensitive properties might be encoded in the lattice of information directly.  The question is how to take an arbitrary property of the form $P \ERarrow Q$ and weaken it to a termination-insensitive variant $P' \ERarrow Q'$.

We are not aware of a general approach to this in the literature. In this section we look at a promising approach  which works for some specific and interesting choices of $P$ and $Q$, but which we failed to generalise. We will later prove that it cannot be generalised in a way which matches the definition which we provide in \S\ref{sec:gtini}.

So how might one weaken a property of the form $P \ERarrow Q$ to allow termination leaks? It is tempting to try to encode this by weakening $Q$ (taking a more liberal relation) -- and indeed that is what has been done in typical relational proofs of simple termination-insensitive noninterference by breaking transitivity and allowing any value in the codomain to be indistinguishable from $\bot$. 
Our approach in \S\ref{sec:gtini} can be seen as a generalisation of this approach.
But it is useful first to consider how far we can get while remaining within the realm of equivalence relations.
 Sterling and Harper in a recent paper on the topic \cite{Sterling:Harper:Sheaf} say (in relation to a specific work \cite{Abadi+:Core} using a relational, semantic proof of noninterference)
\begin{quote}
``A more significant and harder to resolve problem is the fact that the indistinguishability relation \ldots cannot be construed as an equivalence relation''
\end{quote}
While this seems to be true if we restrict ourselves to solving the problem by weakening $Q$, in fact it \emph{is} possible to express termination-insensitivity of types (i) and (ii) just using equivalence relations.
The trick is not to weaken $Q$, but instead to strengthen $P$.

The approach, which we briefly introduce here, is based on Bay and Askarov's study of progress-insensitive noninterference \cite{Bay:Askarov:Reconciling}.  Their idea is to characterise a hypothetical observer who \emph{only} learns through progress or termination behaviour. In the specific case of \cite{Bay:Askarov:Reconciling} it is a ``progress observer''  who sees the length of the output stream, but not the values within it. Let us illustrate this idea in the more basic context of simple termination-insensitive properties.  Suppose we want to define a simple termination-insensitive  variant of a property of the form $f:\ P \ERarrow Q$ for some function $f \in [D \rightarrow V_\bot]$ where $V$ is a flat set of values. We characterise the termination observer by the relation $T = \setof{(\bot,\bot)} \cup \setdef{(\mathrm{lift}(u),\mathrm{lift}(v))}{u \in V, v \in V}$.  The key idea is that we modify the property $f:\ P \ERarrow Q$ not by weakening the observation $Q$, but by strengthening the prior knowledge $P$. We need to express that by observing $Q$ you learn nothing more than $P$ plus whatever you can learn from termination; here ``plus'' means least upper bound, and ``what you learn from termination'' is expressed as the generalised kernel of $f$ with respect to $T$,  namely $f^\ast(T)$.
Thus the simple termination-insensitive weakening of $f:\ P \ERarrow Q$ is
\[
f:\ P \LOIjoin {f^\ast(T)} \ERarrow Q.
\]
The general idea could then be, for each codomain, to define a suitable termination observer $T$.  Bay and Askarov did this for the domain of streams to obtain ``progress-insensitive'' noninterference.  We see two reasons to tackle this differently: 
\begin{enumerate}
    \item Reasoning explicitly about $P \LOIjoin {f^\ast(T)}$ is potentially cumbersome, especially since we don't care \emph{what} is leaked in a termination-insensitive property.
    \item \label{item:non-existence-of-T}
    Finding a suitable $T$ that works as intended but over an arbitrary domain is not only non-obvious, but, we suspect, not possible in general.
\end{enumerate}
In \S\ref{sec:impossibility} we return to point (\ref{item:non-existence-of-T}) to show that it is not possible to find a definition of $T$ which matches the generalised termination-insensitivity which we now introduce.  

\subsection{Using LoCI to Define Generalised Termination-Insensitivity}
\label{sec:gtini}
Here we provide a general solution to systematically weakening an $\LOI$ property $f: R \ERarrow S$ to a termination-insensitive counterpart
(we assume $S$ is realisable).

The first step is to encode $f: {R \ERarrow S}$ as the \LoCI property
$f: P \ERarrow Q$, where $P = \Cp(R)$ and $Q = \Cp(S)$, as allowed by
Corollary~\ref{coroll:encodingLoIinLoCI}.
Preorder $Q$ has the same equivalence classes as $S$, but the classes themselves are minimally ordered to respect the domain order; it is precisely this ordering which gives us a handle on the weakening we need to make.

As a starting point, consider how simple termination-insensitive noninterference is proven: one ignores distinctions that the observer might make between nontermination and termination. In a relational presentation (e.g.~\cite{Abadi+:Core}) this is achieved by simply relating bottom to everything (and vice-versa) and not requiring transitivity.  What is the generalisation to richer domains (i.e.\ domains with more ``height'')?
The first natural attempt comes from the observation that,
in a Scott-style semantics, operational differences in termination behaviour manifest denotationally as differences in definedness, i.e.\ as inequations with respect to the domain ordering.

Towards a generalisation, let us start by assuming that $S$ is the identity, so preorder $Q = {\Cp(\Id)}$ is the top element of \LOCI, i.e.\ it is just the domain ordering. This corresponds to an observer who can ``see'' everything (but some observations are more definite than others).  The obvious weakening of the property $f: P \ERarrow (\poleq)$ is to symmetrise $(\poleq)$ thus:
\[ \setdef{(d,e)}{d \poleq e ~\text{or}~ e \poleq d} \]
This is ``the right thing'' for some domains but not all. As an example of where it does \emph{not} do the right thing, consider the domain $\domTwo \times \domTwo$ where $\domTwo = \domOne_\bot$, and $\domOne = \set{\one}$.  This domain contains four elements in a diamond shape. Suppose that a value of this type is computed by two loops, one to produce the first element, and one to produce the second. A termination-insensitive analysis ignores the leaks from the termination of each loop, so our weakening of any desired relation on $\domTwo \times \domTwo$ must relate $(\bot,\one)$ and $(\one,\bot)$ (and hence termination-insensitivity must inevitably leak all information about this domain).  But what do $(\bot,\one)$ and $(\one,\bot)$ have in common? The answer is that they represent computations that might turn out to be the same, should their computations progress, i.e.\ they have an upper bound with respect to the domain ordering.

What about when the starting point is an arbitrary $Q \in \LOCI(D)$? The story here is essentially the same, but here we must think of the equivalence classes of $Q$ instead of individual elements, and the relation $Q$ instead of the domain ordering.
\begin{definition}[Compatible extension]
  \label{def:consistent}
Given two elements $d, d' \in D$, and a preorder $Q$ on $D$,
we say that $d$ and $d'$ are $Q$-\emph{compatible} if there exists an $e$ such that $d \mathrel{Q} e$ and  $d' \mathrel{Q} e$.
Define $\consi{Q}$, the \emph{compatible extension} of $Q$, to be
$ \setdef{(d,d')}{\text{$d$ is $Q$-compatible with $d'$}} $.
\end{definition}
For any preorder $Q$, compatible extension has the following evident properties:
\begin{enumerate}
    \item ${\consi{Q}} \supseteq {Q}$ (if $d \mathrel{Q} e$ then $e$ is a witness to the compatibility of $d$ and $e$, since $Q$ is reflexive).
    \item ${\consi{Q}}$ is reflexive and symmetric (but not, in general, transitive).
\end{enumerate}
A candidate general notion of termination-insensitive noninterference is then to use properties of the form
\[
    f: P \ERarrow {\consi{Q}}
\]
where $P$ and $Q$ are complete preorders.
This captures the essential idea outlined above,
and passes at least one sanity check:
$f: P \ERarrow {\consi{Q}}$ is indeed a weaker property than
$f: P \ERarrow Q$ (simply because ${\consi{Q}} \supseteq Q$).
However, a drawback of this choice is that it lacks a strong composition property. In general,
${f: P \ERarrow {\consi{Q}}} \wedge {g: Q \ERarrow {\consi{R}}}$
does \emph{not} imply that
${g \circ f}: P \ERarrow {\consi{R}}$.
For a counterexample, consider the following function
$g \in [A \rightarrow A]$,
where $A = \{0,1,2\}_\bot$:
\\[1ex]
\begin{minipage}{0.5\textwidth}
\[
    g(a) = \left\{\begin{array}{cl}
        \bot    & \mbox{ if } a = \bot \\
        0       & \mbox{ if } a = 0 \\
        1       & \mbox{ if } a = 1 \\
        \bot    & \mbox{ if } a = 2
    \end{array}\right.
\]
\end{minipage}
$Q = \mbox{}$ \begin{minipage}[c]{0.4\textwidth}
\includegraphics[width=0.3\textwidth]{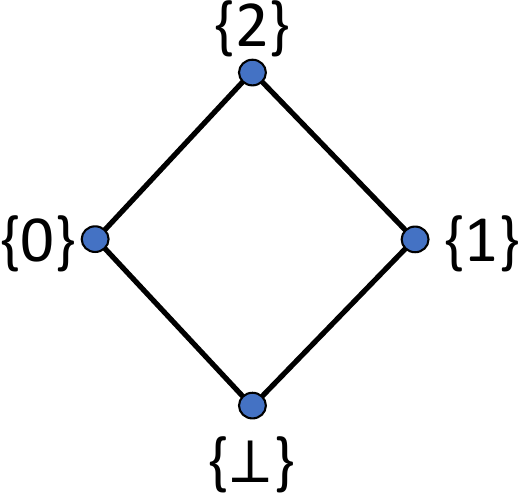}
\end{minipage}
\\[1ex]
Let $Q$ be the complete preorder whose underlying equivalence relation is the identity relation but which orders the elements of $A$ in a diamond shape, as pictured above.
It is easily checked that
$g: Q \ERarrow {\consi{(\poleq)}}$.
Now, since $Q$ has a top element, $\consi{Q}$ is just $\All$,
so for \emph{every} $P$ and $f$ of appropriate type, it will hold that
$f: P \ERarrow {\consi{Q}}$.
But it is \emph{not} true that
${g \circ f}: P \ERarrow {\consi{(\poleq)}}$ holds for every $P$ and $f$
(take $P = \All$ and $f = \id$, for example).

Clearly, the above counterexample is rather artificial. Indeed, it is hard to see how we might construct a program with denotation $g$ such that a termination-insensitive analysis could be expected to verify
$g: Q \ERarrow {\consi{(\poleq)}}$.
Notice that $g$ not only fails to send $Q$-related inputs to $(\poleq)$-related outputs, it effectively ignores the ordering imposed by $Q$ entirely, in that it fails even to preserve $Q$-compatibility.
This suggests a natural strengthening of our candidate notion.
We define our generalisation of termination-insensitive noninterference over the lattice of computable information to be ``preservation of compatibility'':
\begin{definition}[Generalised Termination-Insensitivity]
  \label{def:TIarrow}
  Let $f \in [D \rightarrow E]$ and let
$P$ and $Q$ be elements of $\LOCI(D)$ and $\LOCI(E)$, respectively. Define:
\[ f: P \tiarrow Q \quad \iff \quad f: {\consi{P}} \ERarrow {\consi{Q}} \]
\end{definition}
Crucially, although this is stronger than our initial candidate, it is still a weakening of ${\_ \ERarrow \_}$:
\begin{lemma}
Let $f \in [A \rightarrow B]$. Let $P$ and $Q$ be complete preorders on $A$ and $B$, respectively. Then
\[
    f: P \ERarrow Q \quad \text{implies} \quad f: P \tiarrow Q.
\]
\begin{proof}
Assume $f: P \ERarrow Q$ and suppose $x \mathrel{\consi{P}} y$.
Since $x \mathrel{\consi{P}} y$, there is some $z$ such that
$x \mathrel{P} z$ and $y \mathrel{P} z$.
Since $f: P \ERarrow Q$, we have
${f(x)} \mathrel{Q} {f(z)}$ and ${f(y)} \mathrel{Q} {f(z)}$,
hence ${f(x)} \mathrel{\consi{Q}} {f(y)}$.
\end{proof}
\end{lemma}
Furthermore, Definition~\ref{def:TIarrow} gives us both compositionality and ``subtyping'':
\begin{proposition}
\label{prop:subTI-compTI}
The following inference rules are valid for all continuous functions and elements of $\LOCI$ of appropriate type:\[
\infer[\emph{SubTI}]{\;P' \LoCIgeq P \;\;\; f: P \tiarrow Q \;\;\; Q \LoCIgeq Q'\;}
{\;f: {P'} \tiarrow {Q'}\;}
\hspace{3em}
\infer[\emph{CompTI}]{\;f: P \tiarrow Q \;\;\; g: Q \tiarrow R\;}
{\;g \circ f: {P} \tiarrow {R}\;}
\]
\begin{proof}
We rely on the general Sub and Comp rules (Fact~\ref{fact:sub-comp}).

For SubTI, the premise for $f$ unpacks to
$f: {\consi{P}} \ERarrow {\consi{Q}}$ and the conclusion unpacks to
$f: \consi{P'} \ERarrow \consi{Q'}$.
It suffices then to show that $P' \LoCIgeq P$ implies
${\consi{P'}} \subseteq {\consi{P}}$
(and similarly for $Q, Q')$,
since we can then apply the general Sub rule directly.
So, suppose $P' \LoCIgeq P$,
hence $P' \subseteq P$,
and suppose $x \mathrel{\consi{P'}} y$.
Then, for some $z$, we have $x \mathrel{P'} z$ and $y \mathrel{P'} z$,
thus $x \mathrel{P} z$ and $y \mathrel{P} z$, thus $x \mathrel{\consi{P}} y$,
as required.

For CompTI we observe that it is simply a specialisation of the general Comp rule, since the premises unpack to
$f: {\consi{P}} \ERarrow {\consi{Q}}$ and
$g: {\consi{Q}} \ERarrow {\consi{R}}$, while the conclusion unpacks to
$g \circ f: \consi{P} \ERarrow \consi{R}$.
\end{proof}
\end{proposition}

\subsection{Impossibility of a Knowledge-based Definition}
\label{sec:impossibility}
In this section we return to the question of whether there exists a knowledge-based characterisation which matches our definition of termination-insensitivity, and show why this cannot be the case. 

Suppose we start with an ``ideal'' property of the form $f: P \ERarrow S$, where $S$ is assumed to be realisable (by $\Cp(S)$), and (for simplicity but without loss of generality) $P$ is over a discrete domain (so ${\Cp(P)} = P$).

\newcommand{\lstinmath}[1]{\operatorname
                            {\text{\lstinline[language=haskell]{#1}}}}
\newcommand{\lst}[1]{\ensuremath{\lstinmath{#1}}}
\newcommand{\kKite}{\lstinmath{Kite}}
\newcommand{\kUnit}{\lstinmath{()}}
\newcommand{\kBody}{\lstinmath{Body}}
\newcommand{\kTail}{\lstinmath{Tail}}
\newcommand{\kF}{\lstinmath{f}}
\newcommand{\kFv}{\lstinmath{f'}}
\newcommand{\kG}{\lstinmath{g}}
\newcommand{\kTrue}{\lstinmath{True}}
\newcommand{\kFalse}{\lstinmath{False}}
\newcommand{\kSeq}{\lstinmath{seq}}

The question, which we will answer in the negative, is whether we can  construct a ``termination observer'' $T$ from the structure of the codomain of $f$ such that
\[
f:\ P \LOIjoin {f^\ast(T)} \ERarrow S \iff f:\ P \tiarrow {\Cp(S)}
\]
We build a counterexample based on the following Haskell code:
\begin{lstlisting}[belowskip=-1pt]
  data Kite = Body () () | Tail
\end{lstlisting}
\begin{lstlisting}[belowskip=-1pt]
  spin = spin 
\end{lstlisting}
\begin{lstlisting}
  f h = if h then Body () spin else Body spin ()
  g h = if h then Body () spin else Tail 
\end{lstlisting}

We will use some security intuitions to present the example (since that is the primary context in which termination-insensitivity is discussed). 
Suppose that we view the input to $\kF$ and $\kG$ as either $\kTrue$ or $\kFalse$, and that this is a secret.  We are being sloppy here and ignoring the fact that the input domain is lifted, but that has no consequence on the following. 

\begin{wrapfigure}{r}{0.4\textwidth}\centering
    \includegraphics[width=0.3\textwidth]{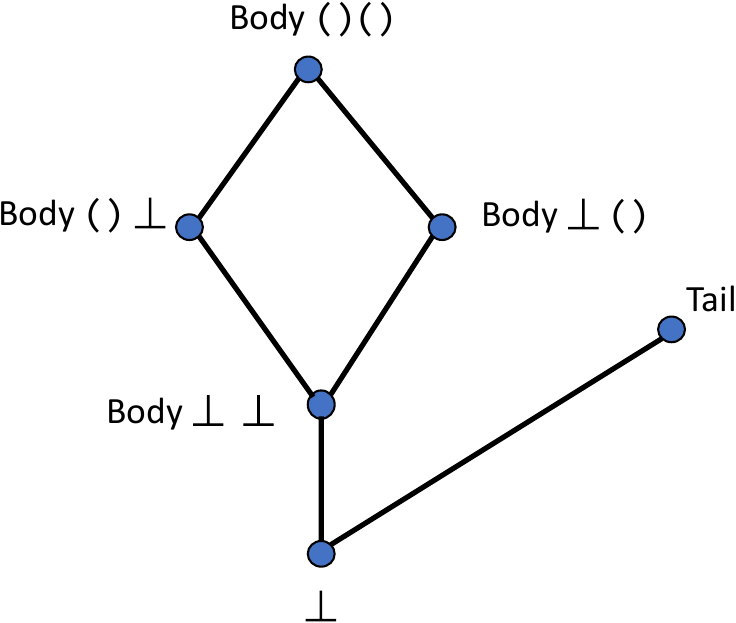}
    \caption{Domain representing \texttt{Kite}}
    \label{fig:kite}
\end{wrapfigure}
Now consider the output to be public, and the question is whether $\kF$ and $\kG$ satisfy termination-insensitive noninterference. Standard noninterference in this case would be the property $(\_ : \All \ERarrow \Id)$. 
Our definition of termination-insensitive noninterference is thus $(\_ : \All \tiarrow (\sqsubseteq))$ where $\sqsubseteq$ here is the ordering on the domain corresponding to \verb!Kite!, namely the domain in \Fig{kite}.
By our definition, $\kF$ satisfies termination-insensitive noninterference but $\kG$ does not. This is perhaps not obvious for $\kF$ because a typical termination-insensitive \emph{analysis} would reject it anyway, so it is instructive to see a semantically equivalent definition $\kFv$ (assuming well-defined Boolean input) which would pass a termination-insensitive analysis\footnote{One should not be surprised that a program analysis can yield different results on semantically equivalent programs -- as Rice's theorem \cite{Rice:Classes:53} shows, this is the price to pay for any non-trivial analysis which is decidable, and having a semantic soundness condition.}.  
\begin{lstlisting}
f' h = Body (assert h ()) (assert (not h) ())
         where  assert b y = seq (if b then () else spin) y 
\end{lstlisting}
 
We claim that a semantic definition of termination-insensitive noninterference should accept $\kFv$ (and hence $\kF$) but reject $\kG$. The reason for this is a fundamental feature of sequential computation, embodied in programming constructs such as call-by-value computation or sequential composition in imperative code.  In Haskell, sequential computation is realised by a primitive function $\kSeq$, which computes its first argument then, if it terminates, returns its second argument.  Consider an expression of the form \lst{seq a b}  where \verb!a! may depend on a secret, but \verb!b! provably does not.  The only way that such a computation reveals information about the secret is if the \emph{termination} of \verb!a! depends on the secret.  This is the archetypal example of the kind of leak that a termination-insensitive analysis ignores. A particular case of this is the function \lst{assert} in the code above, which leaks the value of its first parameter via (non)termination. For this reason, even when \verb!h! is a secret, terms \lst{assert h ()} and 
\lst{assert (not h) ()} are considered termination-insensitive noninterfering (and thus so is $\kFv$). 
This example forms the basis of our impossibility claim, the technical content of which is the following:
\begin{proposition}
There is no termination observer $T$
(i.e.\ an equivalence relation)
on the $\kKite$ domain for which 
$\kF:\ \All \LOIjoin {{\kF}^\ast(T)} \ERarrow \Id$ but for which this does not hold for $\kG$. 
\end{proposition}
\begin{proof}
The problem is to define $T$ in such a way that it distinguishes different $\kBody$ instances but none of the $\kBody$ instances from $\kTail$, while still being an equivalence relation. $T$ would either have to (1) relate $\kBody \kUnit \bot$ and 
$\kBody \bot \kUnit$ or (2) distinguish them and also distinguish one of them from $\kTail$ (if it related both to $\kTail$, then, by transitivity and symmetry, it would also relate the two $\kBody$ instances). Without loss of generality, assume 
$(\kBody \kUnit \bot, \kTail) \not\in T$ 
(otherwise adjust $\kG$ accordingly).
In case (1), $\kF$ does not have property $\All \LOIjoin {{\kF}^\ast(T)} \ERarrow \Id$, because ${\kF}^\ast(T)$ is $\All$, but $\kF \kTrue \neq \kF \kFalse$.
In case (2), $\kG$ does have this property  because ${\kG}^\ast(T)$ is the identity relation.
\end{proof}

\subsection{Case Study: Nondeterminism and Powerdomains}
\label{sec:nondeterminism}
In this section we consider the application of generalised termination-insensitive properties to nondeterministic languages modelled using \emph{powerdomains} \cite{Plotkin:Powerdomain}. 
In the first part we instantiate our definition for a finite powerdomain representing a nondeterministic computation over lifted Booleans and illustrate that it ``does the right thing''.  In the second part we prove that we have an analogous compositional reasoning principle to the function composition property CompTI (Proposition~\ref{prop:subTI-compTI}), but replacing regular composition with the Kleisli composition of the finite powerdomain monad.

\para{Example: Termination-Insensitive Nondeterminism}
\newcommand{\Boolean}{\mathrm{Bool}}
We are not aware of any specific studies of termination-insensitive noninterference for nondeterministic languages, and the definitions in this paper were conceived independently of this example, so it provides an interesting case study. 

\begin{wrapfigure}{r}{0.4\textwidth}
\centering
    \vspace{-10pt}
    \includegraphics[width=0.3\textwidth]{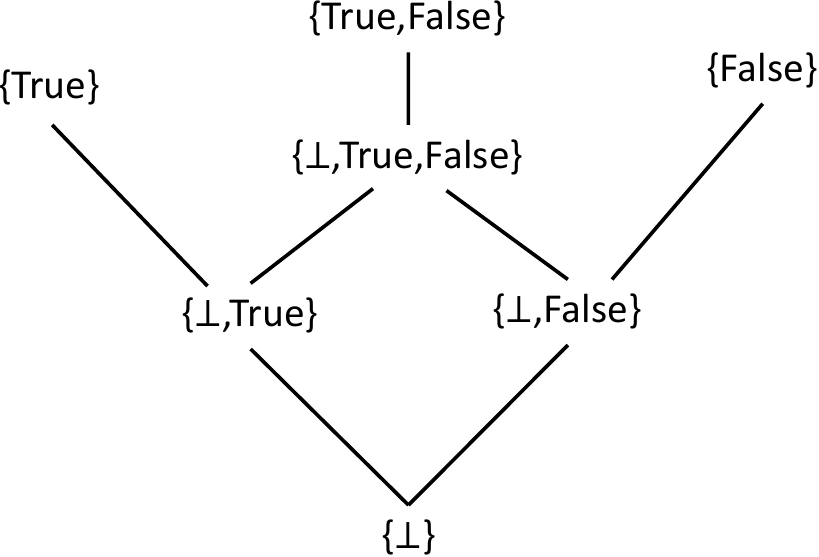}
    \caption{Powerdomain $\powerdomain(\Boolean_\bot)$}
    \label{fig:powerdomainBool}
\end{wrapfigure}
Suppose we have a nondeterministic program $C$, modelled as a function in
\(\Boolean \rightarrow \powerdomain(\Boolean_\bot)\), where \(\powerdomain\) is the Plotkin powerdomain
constructor and \(\Boolean = \set{\kTrue,\kFalse}\).
In the case of powerdomains over finite domains, the elements can be viewed as convex subsets of the underlying domain
(see below for more technical details).
In this section we only consider such finite powerdomains.
\(\powerdomain(\Boolean_\bot)\), for example, is given in Figure~\ref{fig:powerdomainBool}.

Each element of the powerdomain represents a set of possible outcomes of a nondeterministic computation.  Let's consider the input of some program $C$ to be a secret, and the output public.
The property of interest here is what we can call TI-security, i.e.,  \(C : \All \tiarrow (\sqsubseteq)\). 

To explore this property, let us assume an imperative programming language with the following features:
\begin{itemize}
    \item a choice operator $C_1 \NDchoice C_2$ which chooses nondeterministically to compute either $C_1$ or $C_2$,
\item  a Boolean input \texttt{x}, and 
\item an output
statement to deliver a final result.
\end{itemize} 
Note how the semantics of ${\NDchoice}$ can be given as set union of values in the powerdomain. 

Under our definition, the compatible extension of the domain ordering
for \(\powerdomain(\Boolean_\bot)\) relates all the points in the lower diamond to each other. Note that in particular this means that $\set{\bot, \kTrue}$ and $\set{\bot, \kFalse}$ are related.   
This in turn means that the following program $C$ is TI secure:
\begin{lstlisting}
          while True { } | output x
\end{lstlisting}
This looks suspicious, to say the least.
A static analysis would never allow such a program.
But our definition says that it is TI-secure,
since the denotation of $C$ maps
$\kTrue$ to $\set{\bot, \kTrue}$  and
$\kFalse$ to $\set{\bot, \kFalse}$,
and these are compatible by virtue of the common upper bound $\set{\bot, \kTrue, \kFalse}$.

To show that our definition is, nonetheless, ``doing the right thing'',
we can write $C$ in a semantically
equivalent way as:
\begin{lstlisting}
        ( while x { } ; output False ) | ( while (not x) { }; output True)
\end{lstlisting}
Not only is this equivalent, but the insecurity apparent in the first
rendition of the program is now invisible to a termination-insensitive
analysis.

Now we turn to properties relevant to compositional reasoning about generalised termination-insensitivity for nondeterministic programs modelled using finite powerdomains.

\para{Compositional Reasoning for Finite Powerdomains}
\label{sec:powerdomains}
We review the basic theory of finite Plotkin powerdomains,
as developed in \cite{Plotkin:Powerdomain}.
We then define a natural lifting of complete preorders to powerdomains and show that this yields pleasant analogues (Corollary~\ref{corollary:subP-compP})
of the SubTI and CompTI inference rules (Proposition~\ref{prop:subTI-compTI}) with respect to the powerdomain monad. 
Note: throughout this section we restrict attention to finite posets, so a preorder is complete iff it contains the partial order of its domain.

The Plotkin powerdomain construction uses the so-called Egli-Milner ordering on subsets of a poset, derived from the order of the poset.
For our purposes it is convenient to generalise the Egli-Milner definition to arbitrary binary relations:
\begin{definition}[Egli-Milner extension]
  Let $R$ be a binary relation on $A$. Then $\EM(R)$ is the binary relation on subsets of $A$ defined by
  \[ {X \mathrel{\EM(R)} Y} \iff
    {(\forall x \in X.\exists y \in Y. x \mathrel{R} y)
    \wedge
    (\forall y \in Y.\exists x \in X. x \mathrel{R} y)}
  \]
\end{definition}
\begin{fact}
\label{fact:EM}\quad\emph{(1)} $\EM(\_)$ is monotone. \quad\emph{(2)} $\EM(\_)$ preserves reflexivity, transitivity, and symmetry.\end{fact}
The Egli-Milner ordering on subsets of a poset $A$ is then $\EM(\poleq)$.
Note that (2) entails that $\EM(R)$ is a preorder whenever $R$ is a preorder.
However, since antisymmetry is \emph{not} preserved, in general
$\EM(\poleq)$ is only a preorder, so to obtain a partial order it is necessary to quotient by the induced equivalence relation.
Conveniently, the \emph{convex} subsets provide a natural canonical representative for each equivalence class:
\begin{definition}[Convex Closure]
The \emph{convex closure} of $X$ is
$\Cv(X) \eqdef \setdef{ b \in A }{ a \in X, c \in X, a \poleq b \poleq c }$.
\end{definition}
\begin{fact}
\label{fact:convex}
\quad\emph{(1)} $\, \Cv$ is a closure operator.
\quad\emph{(2)} $\, \Cv(X)$ is the largest member of
    $\cell{X}_{\EM(\poleq)}$.
\end{fact}

\begin{definition}[Finite Plotkin Powerdomain]
  Let $(A,\poleq)$ be a finite poset.
  Then the \emph{Plotkin powerdomain} $\powerdomain(A)$ is
  the poset of all non-empty convex subsets of $A$ ordered by $\EM(\poleq)$.
The union operation is defined by
  $X \plotkinunion Y \eqdef {\Cv(X \cup Y)}$.
\end{definition}

The powerdomain constructor is naturally extended to a monad,
allowing us to compose functions with types of the form
$A \rightarrow \powerdomain(B)$.
\begin{definition}[Kleisli-extension]
  Let $A, B$ be finite posets.
  Let $f \in [A \rightarrow \powerdomain(B)]$.
  The Kleisli-extension of $f$ is $\kleisli{f} \in [\powerdomain(A) \rightarrow \powerdomain(B)]$ defined by
$ \kleisli{f}(X) = \Cv(\bigcup_{x \in X} f(x)). $ 
\end{definition}

\begin{definition}[Kleisli-composition]
  Let $A, B, C$ be finite posets and let
  $f \in [A \rightarrow \powerdomain(B)]$ and
  $g \in [B \rightarrow \powerdomain(C)]$.
  Then the Kleisli-composition $f;g \in [A \rightarrow \powerdomain(C)]$ is
  $\kleisli{g} \circ f$.
\end{definition}

We lift the powerdomain constructor to binary relations in the obvious way:
\begin{definition}
  Let $R$ be a binary relation on finite poset $A$.
  Then $\powerdomain(R)$ is the relation on $\powerdomain(A)$ obtained by restricting $\EM(R)$ to non-empty convex sets.
\end{definition}
\begin{lemma}
If $P$ is a complete preorder on finite poset $A$ then
$\powerdomain(P)$ is a complete preorder on $\powerdomain(A)$.
\end{lemma}

Now, in order to establish our desired analogues of SubTI and CompTI, we must be able to relate $\consi{\powerdomain(P)}$ to $\consi{P}$.
The key properties are the following:
\begin{lemma}
\label{lemma:twiddle-EM}
Let $R$ be a preorder and let $P$ be a complete preorder. Then:
\\[1ex]
\mbox{}
\hfill 
\emph{(1)}\quad ${\consi{\EM(R)}} = {\EM(\consi{R})}$
\hfill
\emph{(2)}\quad $\Cv(X) \mathrel{\consi{\powerdomain(P)}} \Cv(Y)$ iff
$X \mathrel{\consi{\EM(P)}} Y$
\hfill
\emph{(3)}\quad ${\consi{\powerdomain(P)}} = {\powerdomain(\consi{P})}$

\end{lemma}
We then have:
\begin{theorem}
Let $A, B$ be finite posets and let
$f \in [A \rightarrow \powerdomain(B)]$.
Let $P, P'$ be a complete preorders on $A$ and let $Q$ be a complete preorder on $B$.
\begin{enumerate}
    \item If $P' \LoCIgeq P$ then ${\consi{\powerdomain(P')}} \subseteq {\consi{\powerdomain(P)}}$.
    \item If $f: P \tiarrow \powerdomain(Q)$ then $\kleisli{f}: \powerdomain(P) \tiarrow \powerdomain(Q)$.
\end{enumerate}
\pagebreak[2] \begin{proof}~
\begin{enumerate}
    \item
        By definition $P' \LoCIgeq P$ iff $P' \subseteq P$ hence,
        as argued in the proof of Proposition~\ref{prop:subTI-compTI},
        $P' \LoCIgeq P$ implies ${\consi{P'}} \subseteq {\consi{P}}$.
        The conclusion then follows by monotonicity of $\EM(\_)$ and
        Lemma~\ref{lemma:twiddle-EM}.
    \item 
        Assume $f: P \tiarrow \powerdomain(Q)$.
        By the definition of $\kleisli{f}$ and Lemma~\ref{lemma:twiddle-EM}, it suffices to show that
        $X_1 \mathrel{\EM(\consi{P})} X_2$
        implies
        $Z_1 \mathrel{\EM(\consi{Q})} Z_2$,
        where $Z_i = \bigcup_{x \in X_i} f(x)$.
Let $z_1 \in Z_1$, thus $z_1 \in f(x_1)$ for some $x_1 \in X_1$.
        Since $X_1 \mathrel{\EM(\consi{P})} X_2$, there is some $x_2 \in X_2$ with
        $x_1 \mathrel{\consi{P}} x_2$.
        Since $f: P \tiarrow \powerdomain(Q)$, it follows that
        $f(x_1) \mathrel{\consi{\powerdomain(Q)}} f(x_2)$, hence by Lemma~\ref{lemma:twiddle-EM}
        $f(x_1) \mathrel{\EM(\consi{Q})} f(x_2)$,
        hence $z_1 \mathrel{\consi{Q}} z_2$ for some $z_2 \in f(x_2) \subseteq Z_2$.
        Thus
        $\forall z_1 \in Z_1.\exists z_2 \in Z_2. z_1 \mathrel{\consi{Q}} z_2$.
        It follows by a symmetrical argument that
        $\forall z_2 \in Z_2.\exists z_1 \in Z_1. z_1 \mathrel{\consi{Q}} z_2$.
        \qedhere
\end{enumerate}
\end{proof}
\end{theorem}

\begin{corollary}
\label{corollary:subP-compP}
Let $A, B, C$ be finite posets and let
  $f \in [A \rightarrow \powerdomain(B)]$ and
  $g \in [B \rightarrow \powerdomain(C)]$.
The following inference rules are valid for all elements of $\LOCI$ of appropriate type:
\begin{gather*}
    \infer{\;P' \LoCIgeq P \;\;\; f: P \tiarrow {\powerdomain(Q)} \;\;\; Q \LoCIgeq Q'\;}
    {\;f: {P'} \tiarrow {\powerdomain(Q')}\;}
    \hspace{3em}
    \infer{\;f: P \tiarrow \powerdomain(Q) \;\;\; g: Q \tiarrow {\powerdomain(R)}\;}
{\;f ; g: {P} \tiarrow {\powerdomain(R)}\;}
\end{gather*}
\end{corollary}

\section{Related Work}
\label{sec:related}
Readers of this paper hoping to see a reconciliation of Shannon's quantitative information theory with domain theory may be disappointed to see that we have tackled a less ambitious problem based on Shannon's lesser-known qualitative theory of information.  \citet{abramsky2008information} discusses the issues involved in combining the quantitative theory of Shannon with the qualitative theory of Scott and gives a number of useful pointers to the literature.  

As we mentioned in the introduction, Shannon's paper describing information lattices \cite{Shannon:Lattice} is relatively unknown, but a more recent account by \citet{Li+:Connection} make Shannon's ideas more accessible (see also \cite{Rioul+:22}). Most later works using similar abstractions for representing information have been made independently of Shannon's ideas.  In the security area, \citet{Cohen:Information} used partitions to describe varieties of information flow via so-called \emph{selective dependencies}. In an independent line of work, various authors developed the use of the lattice of partial equivalence relations (PERs) to give semantic models to polymorphic types in programming languages e.g.~\cite{Coppo:Zacchi:Type,Abadi:Plotkin:PER}. 
PERs generalise equivalence relations by dropping the reflexivity requirement, so a PER is just an equivalence relation on a subset of the space in question. An important generalisation over equivalence relations, particularly when used for semantic models of types, is that ``flow properties'' of the form $f : P \ERarrow Q$ can expressed by interpreting $P \ERarrow Q $ itself as a PER over functions, and  $f : P \ERarrow Q$ is just shorthand for $f$ being related to itself by this PER.
The connection to information flow and security properties comes via \emph{parametricity}, a property of polymorphic types which  can be used to establish noninterference e.g.~\cite{Tse:Zdancewic:Translating,Bowman:Ahmed:Noninterference}. 

Independent of all of the above, \citet{Landauer:Redmond:CSFW93} described the use of the lattice of equivalence relations to describe security properties, dubbing it  
\emph{a lattice of information}.  \citet{Sabelfeld:Sands:PER}, inspired by the use of PERs for static analysis of dependency \cite{Hunt:PhD,Hunt:Sands:PER} (and independent of Landauer and Redmond's work) used PERs over domains to give semantic models of information flow properties, including for more complex domains for nondeterminism and probability, and showed that the semantic properties could be used to prove semantic soundness of a simple type system.  Our TI results in \S~\ref{sec:nondeterminism} mirror the termination-sensitive composition principle for powerdomains given by \citet{Sabelfeld:Sands:PER}. \citet{Hunt:Sands:Quantale} introduce a refinement of \LoI, orthogonal to the present paper, which adds disjunctive information flow properties to the lattice.  \citet{Li+:Downgrading} use a postprocessing definition of declassification policies in the manner of Proposition~\ref{prop:postprocessing}(2); \citet{Sabelfeld:Sands:Declassification} sketched how this could be reformulated within \LOI.

\citet{Giacobazzi:Mastroeni:Abstract2004}
introduced \emph{abstract noninterference} (ANI) in which a security-centric noninterference property is parameterised by abstract interpretations to represent the observational power of an attacker, and the properties to be protected.
\citet{Hunt:Mastroeni:SAS05} showed how so-called \emph{narrow} ANI in \cite{Giacobazzi:Mastroeni:Abstract2004} and some key results can be recast as properties over \LOI.
In its most general form, \citet[Def~4.2]{Giacobazzi:Mastroeni:Abstract:Journal}
define ANI as a property of a function $f$ parameterised by three upper closure operators (operating on \emph{sets} of values):
an output observation $\rho$,
an input property $\phi$ which may flow,
and an input property $\eta$ ``to protect''.
A function $f$ is defined to have abstract noninterference property $\set{\phi,\eta}f\set{\rho}$ if, for all $x, y$:
\[
\phi(\{x\}) = \phi(\{y\}) ~\text{implies}~ \rho(\hat{f}(\eta(\{x\}))) = \rho(\hat{f}(\eta(\{y\}))) 
\]
where $\hat{f}$ is the lifting of $f$ to sets.
Note that this can be directly translated to an equivalent property over the lattice of information, as follows: 
\[
\hat{f} \circ \eta' : \ker(\phi') \ERarrow \ker(\rho)
\]
where
$\phi'(x) = \phi(\{x\})$ and $\eta'(x) = \eta(\{x\})$.
In the special case that $\eta$ is the identity, this reduces simply to an information flow property of $f$, namely:
\[
f: \ker(\phi') \ERarrow \ker(\rho')
\]
where $\rho'(y) = \rho(\{y\})$.
In the general case,
\citet{Giacobazzi:Mastroeni:Abstract:Journal}
observe that the ANI framework models attackers whose ability to make
logical deductions (about the inputs of $f$) is constrained within the abstract interpretation fixed by $\rho$ and $\eta$.
Inheriting from the underlying abstract interpretation framework, ANI can be developed within a variety of semantic frameworks (denotational, operational, trace-based, etc.).

The lattice of information, either directly or indirectly provides a robust baseline for various quantitative measures of information flow \cite{Malacaria:Algebraic,McIver+:Abstract}.
In the context of quantitative information flow, 
\citet[Chapter~16]{Alvim+:Science} discuss leakage refinement orders for potentially nonterminating probabilistic programs.  Their ordering allows increase in ``security'' or termination. Increase in security here corresponds to decrease in information (a system which releases no information being the most secure). Thus \citeauthor{Alvim+:Science}'s ordering is incomparable to ours: the \LoCI ordering reflects increase in information (decrease in security) or increase in termination. 

Regarding the question of termination-sensitive noninterference, the first static analysis providing this kind of guarantee was by \citet{Denning:Certification}.  This used Dorothy Denning's lattice model of information \cite{Denning:Lattice}.  It is worth noting that Denning's lattice model is a model for security properties expressed via labels, inspired by, but generalising, classical military security clearance levels.  As such these are syntactic lattices used to identify different objects in a system and to provide a definition of the intended information flows. But Denning's work did not come with any actual formal definition of information flow, and so their analysis did not come with any proof of a semantic security property.  Such a proof came later in the form of a termination-insensitive noninterference property for a type system \cite{Volpano+:Sound}, intended to capture the essence of Denning's static analysis.  The semantic guarantees for such analysis in the presence of stream outputs was studied by \citet{Askarov+:Termination}.  There they showed that stream outputs can leak arbitrary amounts of information through the termination (progress) insensitivity afforded by a Denning-style analysis, but also that the information is bound to leak slowly.  

In recent work, \citet{Sterling:Harper:Sheaf} propose a new semantic model for termination-insensitive noninterference properties using more elaborate domain-theoretic machinery. The approach is fundamentally type-centric, adopting sheaf semantics ideas from their earlier work on the semantics of phase distinctions \cite{Sterling:Harper:Logical}. 
The fundamental difference in their work is that it is an \emph{intrinsic} approach to the semantics of information flow types in the sense of \citet{Reynolds:What}, whereby information flow specifications are viewed as an integral part of types, and thus the meaning of the type for a noninterfering function is precisely the semantics of noninterference.  This is in contrast with the \emph{extrinsic} relational models studied here in which information flow properties are characterised by properties carved out of a space of arbitrary functions.
Though the merger of domains and relations
sketched in \Sec{categorical} may be considered an intrinsic presentation, the approach of \citeauthor{Sterling:Harper:Logical} goes further: it requires a language of information flow properties to be part of the type language (and 
the underlying semantics). In their work the class of properties discussed is quite specific, namely those specifiable by a
Denning-style (semi)lattice of security labels. The approach is particularly suited to reasoning about systems in the style of DCC \cite{Abadi+:Core} in which security labels are part of the programming language itself.  Unlike in the present work, the only kind of termination-insensitive noninterference  discussed in their paper is the simple case in which a program either terminates or it does not.

The most advanced semantic soundness proof of termination-insensitive noninterference (in terms of programming language features) is the recent work of \citet{Gregersen+:Mechanized}. In terms of advanced typing features (combinations of higher-order state, polymorphism, existential and recursive types\ldots) this work is a tour de force, although the notion of termination-sensitivity at the top level is just the simplest kind; any termination-insensitive notions that arise internally through elaborate types are not articulated explicitly. 

\section{Conclusion and Future Work}
\label{sec:future}
In this paper we have reconciled two different theories of information: 
\begin{itemize}
    \item Shannon's lattice model, which gives an encoding-independent view of the information that is released from some data source by a function, and orders one information element above another when it provides more information about the source; and 
    \item Scott's domain theory, where an ``information element'' is a provisional representation of the information produced so far by a computational process, and the ordering relation reflects an increase in definedness, or computational progress. 
\end{itemize}
Our combination of these models, which we have dubbed the Lattice of Computable Information (as a nod to the fact that Scott's theory is designed to model computable functions via continuity, even if it does not always do so perfectly) retains the essential features of both theories -- it possesses the lattice properties which describe how information can be combined and compared, at the same time as taking into account the Scott ordering in a natural way.  We have also shown how the combination yields the first definition, general in its output domain, of what it means to be the termination-insensitive weakening of an arbitrary flow property.

We identify some lines of further work which we believe would be interesting to explore:

\para{New Information Flow Policies using LoCI} \LoCI allows the expression of new, more fine-grained information flow properties, but which ones are useful?  One example worth exploring relates to noninterference for systems with input streams.  In much of the literature on noninterference for such systems there is an explicit assumption that systems are ``input total'', so that the system never blocks when waiting for a secret input. Using \LOCI we have the machinery to explore this space without such assumptions -- we can formulate what we might call \emph{input termination-sensitive} and \emph{input termination-insensitive} properties within \LOCI (without weakening). Input termination-sensitive properties are very strong since they assume that the high user might try to sneak information to a low observer via the decision to supply or withhold information, whereas insensitive properties permit upfront knowledge of the number of high inputs consumed, thus ignoring these flows.

Another place where \LOCI can prove useful is in \emph{required release} policies \cite{Chong:Required}, where a minimum amount of information flow is required (e.g.\ a freedom of information property).  In this case we would like to ensure that the information which is released is produced in a ``decent'' form -- i.e.\ as a maximal element among those \LOCI elements which have the same equivalence classes.  This prevents the use of nontermination to obfuscate the information. 

\para{Semantic Proofs of Noninterference}
We have developed some basic semantic-level tools for compositional reasoning about information flow properties in \LOCI, and their termination-insensitive relatives
(e.g.\ Proposition~\ref{prop:subTI-compTI} and Corollary~\ref{corollary:subP-compP}).
It seems straightforward to establish a flow-sensitive variant of the progress-insensitive type system of \citet{Askarov+:Termination} that can be given a semantic soundness proof based on the definition of termination-insensitivity given here (the proof in \cite{Askarov+:Termination} is not given in the paper, but it is a syntactic proof
). It would be interesting to tackle a more involved language, for example with both input and output streams, and with input termination-sensitive/insensitive variants.  It would be important, via such case studies, to further develop an arsenal of properties, established at the semantic level, which can be reused across different proofs for different systems.
For languages which support higher-order functions, semantic proofs would call for the ability to build complete preorders on continuous function spaces $[A \rightarrow B]$ by the usual logical relations construction.
That is, given complete preorders $P$ and $Q$ on $A$ and $B$, we would like to construct a complete preorder $P \Rightarrow Q$, relating $f$, and $g$ just when $a \mathrel{P} a'$ implies $f(a) \mathrel{Q} g(a')$.
In fact, defined this way, the relation $P \Rightarrow Q$ will in general only be a \emph{partial} preorder (some elements of $[A \rightarrow B]$ will not be in the relation at all). Promisingly, the results in \cite{Abadi:Plotkin:PER} suggest that complete partial preorders are well-behaved, yielding a cartesian-closed category.

\para{Domain Constructors}
The powerdomain results in \S\ref{sec:powerdomains} are limited to finite posets.
It would be interesting to extend these results beyond the finite case and, more generally, to see if other domain constructions, including via recursive domain equations, can be lifted to complete preorders.
Clearly this will require restriction to an appropriate category of algebraic domains, rather than arbitrary posets.
It remains to be seen whether it will also be necessary to impose additional constraints on the preorders.

\begin{acks}
Thanks to the anonymous referees for numerous constructive suggestions, in particular connections to category theory that formed the basis of \S\ref{sec:categorical}, and the suggestion to use an example based on powerdomains. Thanks to Andrei Sabelfeld and Aslan Askarov for helpful advice.  This work was partially supported by the Swedish Foundation for Strategic Research
(SSF), the Swedish Research Council (VR). 
\end{acks}
\vfill
\pagebreak
\bibliography{bib-bibliography}

\end{document}